%% file: Main.tex
\documentclass[12pt]{article}
\input{Macros.tex}

\title{Good Quantum Locally Testable Codes from Lossless Cubical Complexes}

\author[1]{Itay Cohen\thanks{Supported by ERC starting grant 949499. \\ Email: {\tt itay3@mail.tau.ac.il}.}}
\author[1]{Itai Leigh\thanks{Supported by the European Union (ERC, EACTP, 101142020), and the Danish-Israeli Study Foundation in Memory of Josef and Regine Nachemson.  \\Email: {\tt itai.leigh@mail.huji.ac.il}.}}
\author[2]{Assaf Reiner\thanks{Supported by the European Research Council (ERC) under the European Union’s Horizon 2020 (N. 882751). Email: {\tt assaf.reiner@mail.huji.ac.il}.}}
\author[1]{Amnon Ta-Shma\thanks{Funded by the Israel Science Foundation grant 443/22. \\ Email: {\tt amnon@tauex.tau.ac.il}.}}
\author[3]{\\ Elad Tzalik\thanks{Supported by the Adams Fellowship Program of the Israel Academy of Sciences and Humanities. \\ Email: {\tt elad.tzalik@weizmann.ac.il}.}}

\affil[1]{Department of Computer Science, Tel Aviv University, Israel.}
\affil[2]{Department of Mathematics, Hebrew University of Jerusalem, Israel.}
\affil[3]{Department of Mathematics and Computer Science, Weizmann Institute of Science, Israel. }
\date{}

\begin{document}
\maketitle

\begingroup
\renewcommand{\thefootnote}{}
\footnotetext{Any views, opinions, findings, and conclusions or recommendations expressed in this material are those of the author(s) and do not necessarily reflect the views of the European Union or the European Research Council Executive Agency. Neither the European Union nor any granting authority can be held responsible for them.}
\endgroup

\begin{abstract}
Sipser and Spielman constructed LDPC codes from either bipartite \emph{spectral} expanders or one-sided \emph{lossless} expanders. 
In higher dimensions, \emph{spectral} expansion similarly played a central role in the constructions of asymptotically good classical LTCs and qLDPC codes by Dinur, Evra, Livne, Lubotzky, and Mozes and by Panteleev and Kalachev. 
Alternatively, Lin and Hsieh constructed classical LTCs and qLDPC codes from two-dimensional \emph{lossless} cubical complexes.

In this work we develop the higher-dimensional \emph{lossless} approach. 
We do not construct the required high-dimensional lossless cubical complexes; rather, we investigate what their existence would imply.
We associate with a high-dimensional cubical complex a \emph{level chain complex}, whose chain groups are supported on the level sets of the Boolean cube rather than on its cells.
Our main technical contribution is a clean local-to-global theorem for this structure: suitable one-dimensional lossless expansion in the directional graphs implies small-set coboundary expansion of the global level complex.
As a consequence, sufficiently imbalanced, two-sided lossless four-dimensional cubical complexes give rise to asymptotically good quantum locally testable codes.
We expect the local-to-global principle developed here to have further applications.

% Independently of our work, asymptotically good qLTCs were recently constructed using \emph{spectral} high-dimensional cubical complexes, and \emph{cellular} chain complexes.
% % We view the spectral and lossless approaches as complementary, much as in the classical expander-code setting: each offers advantages that are not naturally captured by the other.
\end{abstract}
\clearpage
\tableofcontents
\clearpage

\input{Introduction}
\input{Cubical-test}
\input{Local}
\input{GTL}

\input{Qltc}

\bibliographystyle{alpha}
\bibliography{Bib}

%\appendix

%\input{HsiehLin}
\end{document}

%% file: Macros.tex
\usepackage{graphicx}
\usepackage{url}
\usepackage[rgb]{xcolor}
\usepackage{verbatim}
\usepackage{amsgen,amsmath,amstext,amsbsy,amsopn,amssymb,amsthm}
\usepackage{fancyhdr}
\usepackage{tikz}
\usepackage{float}
\usepackage{algorithm}
\usepackage{algpseudocode}
\usepackage{subcaption}
\usetikzlibrary{positioning, fit}
\usetikzlibrary{arrows.meta,shapes.geometric}

\usepackage{rotating}

\usepackage{mathtools}
\usepackage{nicefrac}

\usepackage[margin=0.85in]{geometry}

\usepackage{thmtools}
\usepackage[backref,colorlinks=true,linkcolor=blue,citecolor=violet]{hyperref}
\usepackage[nameinlink, noabbrev, capitalize]{cleveref}
\hypersetup{hypertexnames=false} %fixing mistake in Cref linking (Theorem 3.2 \label{tlg} somehow was linked incorrectly)

\usepackage[utf8]{inputenc}
\usepackage[english]{babel}
\IfFileExists{authblk.sty}{%
  \usepackage{authblk}%
}{%
  \newcounter{Maxaffil}%
}
\usepackage{setspace}

\declaretheorem[
  name=Theorem,
  numberwithin=section
]{theorem}

\declaretheorem[
  name=Definition,
  sibling=theorem
]{definition}

\declaretheorem[
  name=Lemma,
  sibling=theorem
]{lemma}

\declaretheorem[
  name=Claim,
  sibling=theorem
]{claim}

\declaretheorem[
  name=Proposition,
  sibling=theorem
]{proposition}

\declaretheorem[
  name=Corollary,
  sibling=theorem
]{corollary}

\declaretheorem[
  name=Remark,
  sibling=theorem
]{remark}

\numberwithin{equation}{section}

\newcounter{Chapcounter}

\newcommand{\chapter}[1] 
{ {\centering          
  \addtocounter{Chapcounter}{1} \Large \underline{\textbf{ \color{blue} Chapter \theChapcounter: ~#1}} }   
  \addcontentsline{toc}{section}{ \color{blue} Chapter:~\theChapcounter~~ #1}    
}

\newcommand{\eps}{\varepsilon}
\newcommand{\set}[1]{\left\{#1\right\}}% parentheses
\newcommand{\prn}[1]{\left(#1\right)}% parentheses
\newcommand{\brk}[1]{\left[#1\right]}% brackets
\newcommand{\norm}[1]{\left\lVert#1\right\rVert}% norm
\newcommand{\abs}[1]{\left\lvert#1\right\rvert}% absolute value

\newcommand{\defeq}{\mathrel{:=}}

\newcommand{\Ftwo}{\mathbb{F}_2}
\newcommand{\meet}[1][]{%
    \if\relax\detokenize{#1}\relax
        \wedge
    \else
        \mathbin{\overset{#1}{\wedge}}%
    \fi
}

\newcommand{\Ord}{\operatorname{Ord}}
\newcommand{\End}{\operatorname{End}}
\newcommand{\Ter}{\operatorname{Ter}}
\newcommand{\Ch}{\operatorname{Ch}}
\newcommand{\Ex}{\operatorname{E}}
\newcommand{\Same}{\operatorname{Same}}
\newcommand{\Sametype}{\operatorname{SameType}}
\newcommand{\Light}{\operatorname{Light}}
\newcommand{\Lightsquare}{\operatorname{LightSquare}}

\newcommand{\type}{\operatorname{type}}

\DeclareMathOperator{\dist}{dist}

\DeclareMathOperator{\im}{im}

\DeclareMathOperator{\Lk}{Lk}
\DeclareMathOperator{\supp}{supp}
\DeclareMathOperator{\Roots}{Rts}

\newcommand{\M}{\mathsf M}

\newcommand{\calC}{\mathcal{C}}
\newcommand{\FF}{\mathbb{F}}

%% file: Introduction.tex
\section{Introduction}
\label{sec:introduction}

Error-correcting codes enable reliable communication and storage in the presence of noise by encoding information into longer, redundant words. A linear code can be described either by a generating matrix or by a parity-check matrix. Although these two descriptions define the same code, they correspond to fundamentally different computational tasks. The generating matrix maps an information message to its encoded codeword. The parity-check matrix, on the other hand, acts on a received, possibly corrupted, word, producing its syndrome and enabling efficient  verification, decoding, and local correction.

These two representations exhibit a fundamental asymmetry. Every column of a generating matrix is itself a codeword, as it is the image of a standard basis vector. Consequently, if the code has linear distance, every nonzero column must have linear Hamming weight, ruling out sparse generating matrices. Sparse parity-check matrices, in contrast, are fully compatible with good code parameters. This observation led Gallager to introduce \emph{low-density parity-check (LDPC) codes}~\cite{Gal62}, whose parity-check matrices have constant row and column weight. Since their introduction, LDPC codes have become one of the central objects in coding theory, combining excellent asymptotic parameters with efficient decoding algorithms~\cite{SS96}.

The sparsity of the parity-check matrix turned out to be important for much more than efficient decoding. Because each parity check involves only constantly many symbols, it enables algorithms that inspect only a small portion of the received word. Beginning with Gallager's iterative decoding algorithms~\cite{Gal62}, this observation evolved into a broad theory of \emph{local} algorithms for error-correcting codes. Depending on the task, one seeks to locally correct individual symbols of a corrupted codeword, locally decode the underlying message, locally test whether a word belongs to or is close to the code, or infer global properties of a word from a small number of queries. These notions have become central throughout coding theory and have found numerous applications in complexity theory, property testing, probabilistically checkable proofs, data structures, and cryptography.

Among these local tasks, local testability occupies a particularly distinguished position. A code is \emph{locally testable} if there exists a randomized algorithm that queries only constantly many coordinates of the input word and accepts every codeword with probability one, while rejecting words with probability proportional to their distance from the code. Local testability is one of the central themes underlying PCP constructions. Goldreich and Sudan isolated this notion in the form of locally testable codes (LTCs)~\cite{GS02}, and it is a key ingredient in Dinur's gap-amplification proof of the PCP theorem~\cite{Din07}. Dinur further showed how to construct constant-query locally testable codes encoding $n$ information bits into $n\operatorname{polylog}(n)$, leaving open whether it is possible to eliminate the polylogarithmic overhead. This question was resolved only recently by the breakthrough constructions of Panteleev and Kalachev~\cite{PK22} and independently Dinur, Evra, Livne, Lubotzky and Mozes~\cite{DELLM26}, who constructed asymptotically good $c^3$-locally testable codes (the $c^3$ indicates all three parameters are constants: rate, relative distance and number of queries). Nevertheless, the existence of linear-length PCPs remains one of the central open problems in complexity theory.

In a seemingly unrelated direction, the quest for scalable quantum computation has motivated the development of quantum error-correcting codes. Since the pioneering works of Shor~\cite{Sho95} and Steane~\cite{Ste96} on quantum error correction, and the threshold theorems of Aharonov and Ben-Or~\cite{AB97} and of Knill, Laflamme, and Zurek~\cite{KLZ96} establishing the possibility of fault-tolerant quantum computation, quantum error correction has been regarded as one of the principal ingredients for building scalable quantum computers. A particularly influential milestone was Kitaev's toric code~\cite{Kit03}, which demonstrated that quantum information can be protected using local stabilizer measurements arranged on a two-dimensional lattice, introducing both topological quantum codes and geometric locality as guiding principles. The toric code, however, has vanishing rate, motivating the search for quantum codes with sparse stabilizer checks and \emph{asymptotically good} parameters. This led to the study of \emph{quantum low-density parity-check (qLDPC) codes}, the quantum analogue of classical LDPC codes. Following a sequence of breakthrough constructions, it is now known that qLDPC codes can simultaneously achieve constant rate, linear distance, and sparse stabilizer checks, resolving a longstanding open problem in quantum coding theory.

Beyond efficient quantum error correction, one may ask for stronger local robustness properties. Quantum locally testable codes (qLTCs), introduced by Aharonov and Eldar~\cite{AE17}, require that every low-energy state is close to the code space, providing a quantum analogue of classical local testability. Such robust local detection has attracted considerable attention because of its potential connections to self-correcting quantum memories, the quantum PCP conjecture, and fault-tolerant architectures with passive error suppression. The characterization of quantum local testability developed by Aharonov, Eldar, and collaborators naturally suggests a higher-dimensional homological framework, as we next explain.

\subsection{The gradual emergence of expansion and homology in coding theory}

Although the classical and quantum theories developed largely independently, a common mathematical picture gradually emerged. Expansion emerged as the combinatorial mechanism that transforms sparse local constraints into global distance and robustness guarantees, while homology emerged as the geometric framework that organizes these constraints and their interactions. 
We next review how such ideas evolved through the successive development of LDPC codes, LTCs, qLDPC codes, and qLTCs.

Gallager's original work already suggested that the combinatorial structure of the Tanner graph plays a decisive role in the performance of an LDPC code. A natural first principle was to require the graph to have large girth, ensuring that local neighborhoods remain tree-like and preventing short cycles from creating spurious correlations during iterative decoding. While this viewpoint proved highly successful in practice, it could not by itself explain the global distance and robustness properties exhibited by non-explicit LDPC codes.

The decisive advance came with the work of Sipser and Spielman~\cite{SS96}, who identified combinatorial expansion as the property underlying good LDPC codes. They showed that if every sufficiently small set of variable nodes has almost the maximum possible number of neighboring check nodes, then the resulting LDPC code has linear distance and admits efficient decoding. This theorem identified a remarkably clean combinatorial mechanism explaining why sparse local constraints can nevertheless enforce strong global properties. At the time, however, such explicit lossless expanders with the required parameters were not known. Sipser and Spielman~\cite{SS96} therefore complemented their combinatorial theorem with an elegant analysis based on spectral expanders, showing that comparable coding-theoretic guarantees can also be derived from spectral expansion assumptions.

The work of Sipser and Spielman~\cite{SS96} illustrates a recurring theme in the development of coding theory where the identification of a clean combinatorial principle often precedes its realization. One first isolates the combinatorial property underlying a coding-theoretic phenomenon, and only later learns how to construct objects satisfying this property explicitly, or to derive comparable conclusions from weaker assumptions using more sophisticated techniques. 
The study of lossless expanders provides a striking illustration of this pattern. Random bipartite graphs establish the existence of optimal lossless expanders by probabilistic arguments, while explicit constructions proved considerably more challenging. Over the following decades, a substantial body of work developed explicit constructions of balanced and unbalanced bipartite lossless expanders, 
eventually realizing the combinatorial objects that had originally motivated the theory~\cite{CRVW02,HLMRZ25}. 

Local testability introduces a phenomenon that does not arise in ordinary LDPC codes. A simple thought experiment illustrates the difficulty. Consider an LDPC code defined by a full-rank parity-check matrix whose Tanner graph is a good expander. Removing a single parity check enlarges the code by a factor of two, while leaving the expansion, and therefore the linear distance, essentially unchanged. Consequently, every new codeword lies far from the original code, yet violates only the deleted parity check. Such a word is globally far from the code but locally appears almost valid, showing that independent parity checks cannot support local testability. This intuition was formalized by Ben-Sasson, Guruswami, Kaufman, Sudan, and Viderman~\cite{BSGKSV09}, who proved that locally testable codes necessarily require substantial redundancy among their checks. 

From an algebraic perspective, this redundancy means that the parity-check matrix is no longer an isolated linear map. Besides describing the code as its kernel, one must also understand the linear relations among its rows. Equivalently, the image of the parity-check matrix must itself satisfy another family of linear constraints. Rather than a single linear map, one is naturally led to a sequence of linear maps
$C_2 \xrightarrow{\partial_2} C_1 \xrightarrow{\partial_1} C_0$,
where the image of each map is contained in the kernel of the next. This is precisely the setting of chain complexes.

This homological viewpoint emerged only gradually. Initially, an LDPC code is naturally described by its Tanner graph. From today's perspective, however, the code associated with a Tanner graph is simply a one-dimensional chain complex. This interpretation is easy to overlook, and perhaps rightly so, because in one dimension the homological language adds little beyond the graph itself. The graph was always present; what emerged only later was the realization that it belongs to a hierarchy of increasingly rich chain complexes. Only after two-dimensional constructions entered coding theory did it become clear that classical LDPC codes were already the first instance of a broader homological picture.

This perspective proved remarkably fruitful. Dinur, Evra, Livne, Lubotzky, and Mozes~\cite{DELLM26} constructed asymptotically good locally testable codes using left-right Cayley complexes, which can be viewed as a special class of two-dimensional cubical complexes. Here, homological expansion in the middle layer of the chain complex provides the redundancy between constraints required for local testability. Around the same time, the breakthrough constructions of Panteleev and Kalachev~\cite{PK22}, together with subsequent developments, demonstrated that closely related two-dimensional structures also suffice to construct asymptotically good quantum LDPC codes.

Following these developments, Lin and Hsieh~\cite{HL22a,HL22b} applied the balanced product construction of Breuckmann and Eberhardt~\cite{BreuckmannEberhardt2021} to lossless expanders, establishing the homological expansion properties underlying their classical LTCs and quantum LDPC codes. A particularly appealing feature of their approach is that in both settings their results share the same combinatorial object, and the same expansion argument. Viewed in the language of lossless square complexes, their results provide the basis for the unified framework we develop here.

The distinction lies in where the information is stored. In the classical construction, the bits occupy one end of the chain complex, the adjacent boundary map defines the parity checks, and the second boundary map provides the checks on checks required for local testability. In the quantum construction, the two ends of the chain complex correspond to the $X$- and $Z$-stabilizer, with the qubits sitting between them.  Correspondingly, the classical construction requires one-sided lossless expansion, whereas controlling both quantum distances naturally leads to two-sided lossless expansion. 

The natural higher-dimensional analogue of a square complex is a cubical complex.
Explicit higher-dimensional cubical complexes with strong spectral expansion are already known and have led to the construction of lossless vertex expanders mentioned before. 
The characterization of quantum local testability developed by Aharonov, Eldar, and collaborators naturally suggests extending the homological picture one step further. Although quantum local testability is defined for arbitrary stabilizer codes, the CSS setting admits a particularly natural realization in which one introduces an additional layer of relations for the $X$-stabilizer checks and another one for the $Z$-stabilizer checks. This naturally points toward higher-dimensional cubical complexes equipped with suitable local expansion properties. Indeed, Dinur, Lin and Vidick~\cite{DLV24} constructed quantum locally testable codes with polylogarithmic overhead this way.

In contrast, higher-dimensional cubical complexes satisfying suitable two-sided lossless expansion are not currently known, even by probabilistic arguments. This situation closely parallels the earlier development of LDPC codes and LTCs. In both cases, a clean combinatorial theorem based on lossless expansion preceded the explicit construction of the corresponding combinatorial object.
The present work continues this pattern in higher dimension. We show that 

\begin{center}
{\bf higher-dimensional cubical complexes with suitable two-sided lossless expansion give rise to asymptotically good quantum locally testable codes}
\end{center}
and leave open the challenge of proving that these combinatorial objects exist and constructing them explicitly.

\subsection{The combinatorial framework}

The basic combinatorial object in this paper is a \emph{decorated cubical complex}. Such a complex in dimension $k$ consists of vertex sets
\[
\{V_I\}_{I\subseteq[k]},
\]
indexed by the subsets of $[k]$, together with undirected 
%direction 
graphs connecting
\[
V_I \qquad\text{and}\qquad V_{I\cup\{i\}},
\]
whenever $i\notin I$. We call each set $V_I$ a block. Thus a decorated cubical complex is obtained from the Boolean cube by replacing each vertex with a block and each edge with a bipartite graph between the corresponding blocks. We refer to the edges between $V_I$ and $V_{I\cup\{i\}}$ as the \emph{direction-$i$ edges}.

A natural chain complex associated with a cubical complex is its
\emph{cellular} chain complex, whose degree-$i$ space is spanned by
the $i$-dimensional cubes. Constructions organized around cells,
with coefficient spaces supplied by local codes, appear
in~\cite{DELLM26,DLV24,GJ26}. In this work, however, we use the
\emph{level} chain complex, extending the two-dimensional framework
of~\cite{PK22,HL22a,HL22b}. We write
\[
    V_i=\bigcup_{|I|=i}V_I,
\]
and call $V_i$ the \emph{level} $i$ of the complex. The level chain complex has $\Ftwo^{V_i}$ as its $i$th vector space and its boundary and coboundary maps are induced by adjacency between consecutive levels. To ensure that successive coboundary maps compose to zero, we require the \emph{weak completion property}, which asserts that for each pair of edges in two distinct directions coming out of a vertex, there exists a unique square in the complex containing them. The weak completion property ensures that successive boundary maps compose to zero: each two-step path is paired with the other path around its unique
square, and the two contributions cancel over $\Ftwo$.

% We consider two notions of cubicality. The \emph{weak cubical property} requires that for every vertex $v\in V_i$, 
% every choice of $r$ incident edges from $v$, one in each of $r$ distinct directions, extends to a unique $r$-cube. The \emph{strong cubical property} additionally requires that every pair of vertices
% $v\in V_i$ and $v'\in V_{i+r}$ connected by a monotone (level increasing) path of length $r$
% are the opposite vertices of a unique $r$-cube. The weak property suffices for our analysis.

The simplest examples of cubical complexes are left-right Cayley complexes. Here there are a (possibly non-Abelian) group $\Gamma$ and four vertex sets
\[
V_\emptyset=V_{\set{1}}=V_{\set{2}}=V_{\set{1,2}}=\Gamma,
\]

\input{figures/left-right-cayley}

In the left-right complex, the direction-$1$ edges are defined by left multiplication by generators from a symmetric set $A$, while the direction-$2$ edges are defined by right multiplication by generators from a symmetric set $B$. For every $g \in V_\emptyset$, every choice of a direction-$1$ edge and a direction-$2$ edge leaving $g$ leads uniquely to
\[
ag\in V_{\set{1}},\qquad
gb\in V_{\set{2}},\qquad
agb\in V_{\set{1,2}},
\]
which form a unique square $\set{agb;ag,gb;g}$, and these are all the squares. This is precisely the weak completion property. We remark that unlike this group example, in a general decorated cubical complex the vertex sets $V_I$ need not have the same cardinality.  

Throughout this paper we assume that every direction graph between
$V_I$ and $V_{I\cup\{i\}}$
is a lossless expander. 
We call such a decorated cubical complex \emph{direction-wise lossless}.
Since every pair of blocks $V_I$, $V_J$ with $I \subseteq J$
is connected
by a monotone sequence of direction graphs, the corresponding lossless expansion estimates compose to give analogous expansion estimates between every such pair of blocks.
It is therefore natural to ask whether the resulting cubical complex itself is an expanding object, and if so, in what way.

A first attempt would be to require every sufficiently small cochain to have a large coboundary. This, however, is impossible. Since
$\delta^2=0$,
every coboundary has zero coboundary. Thus any meaningful notion of expansion must first quotient out these trivial obstructions. We say
\begin{itemize}
\item A cochain in $V_i$ is \emph{minimal} if it has minimum norm among all cochains in its cohomology class, or equivalently, if adding the coboundary of an arbitrary $(i-1)$-cochain cannot reduce its norm.
\item 
A cochain is \emph{locally minimal} 
if no local coboundary operation supported in the neighborhood of 
a single link 
can decrease its norm.
\end{itemize}
Thus minimality is a global notion, whereas local minimality can be verified locally.

We say that a decorated cubical complex is a \emph{small-set level coboundary expander} if every sufficiently small minimal cochain supported on a level below the top has large coboundary. 
The terminology reflects the fact that ordinary vertex expansion takes place between adjacent blocks, whereas small-set level coboundary expansion measures expansion between adjacent levels, after quotienting out the trivial cohomological obstructions. Similarly, we say the cubical complex is a \emph{robust}\footnote{Outside the introduction, we omit the adjective ``robust'': small-set coboundary expansion always refers to expansion of locally minimal cochains.} small-set level coboundary expander if every sufficiently small \emph{locally} minimal cochain has large coboundary.

\subsection{The result}

The central combinatorial statement of this paper is the following theorem:

\begin{theorem}[Informal version of \Cref{prop:small-set-cosystolic-expansion}]
\label{thm:cosys:expansion:informal}
Every direction-wise lossless decorated cubical complex satisfying suitable lossless expansion and imbalance assumptions is a robust small-set level coboundary expander.  
\end{theorem} 

This implication is far from obvious. The hypothesis concerns only the direction graphs of the complex, whereas the conclusion is a genuinely higher-dimensional expansion property. 

\Cref{thm:cosys:expansion:informal} is the higher-dimensional analogue of the two-dimensional expansion theorem underlying the constructions of Lin  and Hsieh~\cite{HL22a,HL22b}. We establish the corresponding combinatorial statement in arbitrary dimension $k$.

Robust small-set level coboundary expansion has the following coding-theoretic consequences.
In dimension two, robust small-set level coboundary expansion immediately yields LTCs, recovering the square-complex constructions of Lin and Hsieh~\cite{HL22a}. The same argument, applied to the very same square complex but interpreted as a CSS chain complex, yields quantum LDPC codes~\cite{HL22b}. 
More generally, following the argument of Aharonov and Eldar \cite{AE17}, a four-dimensional robust small-set level coboundary expander gives rise to a quantum locally testable code whenever the corresponding chain complex is interpreted as a CSS code.
Thus,
\[
\boxed{\text{robust small-set level coboundary expander}}
\quad\Longrightarrow\quad
\begin{cases}
\text{LTC}, & k=2,\\[2mm]
\text{qLDPC}, & k=2,\\[2mm]
\text{qLTC}, & k=4.
\end{cases}
\]

Combining this implication with Theorem~\ref{thm:cosys:expansion:informal}, we obtain
\[
\begin{gathered}
\boxed{\text{Direction-wise lossless decorated cubical complex + imbalance conditions}}\\[2mm]
\Downarrow\\[2mm]
\boxed{\text{Robust small-set level coboundary expander}}\\[2mm]
\Downarrow\\[2mm]
\boxed{\text{qLTC}}
\end{gathered}
\]
The remaining challenge  is therefore purely combinatorial: construct explicit higher-dimensional direction-wise lossless decorated cubical complexes.

\paragraph{A candidate construction, and an obstruction.}
Explicit constant-degree lossless vertex expanders were constructed by
Hsieh et al.~\cite{HLMRZ25}, but it is unclear how to adapt their construction
to obtain higher-dimensional cubical complexes with lossless directional
graphs. A natural starting point is the family of balanced cubical complexes
underlying their construction, obtained as quotients of products of
$(p_i+1)$-regular Bruhat--Tits trees by congruence subgroups of quaternionic
arithmetic lattices. Their directional graphs inherit the Ramanujan property
from the graphs of Lubotzky, Phillips, and Sarnak~\cite{LPS88}.
Our setting needs imbalance, so we set out to construct the imbalanced analogues of those complexes. The natural input is the family of Ramanujan \emph{bigraphs} recently constructed by Evra, Feigon, Maurischat and Parzanchevski~\cite{EFMP23,EFMP25}, which are congruence quotients of the $(p^3+1,p+1)$-biregular Bruhat-Tits tree of a definite unitary group in three variables. Replacing the quaternionic lattice by the unitary one, and the product of regular trees by a product of biregular ones, one obtains for every $k$ a strongly explicit $k$-dimensional imbalanced decorated cubical complex with complete multipartite links, whose directional graphs are all Ramanujan bigraphs of degrees $(p_i^3+1,p_i+1)$, and whose degrees may be imbalanced as strongly as our imbalance assumption requires. Such a complex satisfies every hypothesis of \cref{thm:cosys:expansion:informal} except one, the lossless expansion of its directional graphs.

The hope that this last hypothesis holds stems from the heuristic that led Kamber and Kaufman~\cite{KK22} to conjecture that the LPS graphs are lossless~\cite[Conjecture 1.11]{KK22}. Pure spectral arguments cannot establish losslessness, as demonstrated by Kahale~\cite{Kah95} in the regular case and by Asherov and Dinur~\cite{AD24} in the biregular setting. However, the known counterexamples of Kamber and Kaufman, namely the Morgenstern graphs, fail due to an obstruction with no counterpart over $\mathbb{Q}$, which is the existence of a proper closed subfield of the local field. Because $\mathbb{Q}_p$ has no proper closed subfields, the same reasoning supports both the losslessness of LPS graphs and that of the EFMP bigraphs, which are likewise defined over $\mathbb{Q}$.

It turns out that the opposite is the case. Adapting the closed orbit method of~\cite{KK22} to the unitary setting, one finds in every EFMP bigraph a nonempty set that does not expand at all. The set arises as the image of a subtree attached to an algebraic subgroup, rather than to a subfield as in~\cite{KK22}. All in all, the EFMP bigraphs, and with them the complexes described above, are not two-sided lossless. These results appear to be of independent interest, and we intend to present them separately rather than develop them here.

The remainder of the introduction is devoted to explaining why the combination of a cubical complex with one-dimensional lossless expansion along the direction graphs suffices to establish this genuinely higher-dimensional expansion theorem.

\subsection{The technique}

The main technical contribution of this work is a dimension-independent global-to-local reduction, building on ideas introduced by~\cite{FK24,EK24,KKL16}. The key idea is that a global cochain should be understood through its local views. The first question is therefore: what is the correct notion of a local view in a decorated cubical complex?

The answer is the notion of a \emph{link}.
%The natural local object in a decorated cubical complex is the link of a face. 
The link of a vertex $x$ consists of $x$ and all vertices in higher levels that are connected to $x$ by a monotone path, together with their incidence structure. It is therefore the smallest upward closed subcomplex containing the vertex $x$. Restricting a cochain to the link of a vertex $x$ captures precisely its local behavior around that vertex $x$.
Consequently, links are the natural local structures into which the input cochain should be decomposed.

Then, given a cochain
\[
S:V_i\rightarrow\Ftwo
\]
of norm at most $\alpha$, we decompose
$S$ into a collection of cochains $L_t:V_i\rightarrow\Ftwo$ that are local in the sense that each of them is supported on an individual link. Our global-to-local theorem
constructs a local decomposition
\[
S \rightsquigarrow \{L_t\}_{t\in T},
\]
such that the global coboundary is lower bounded, up to a small error, by the sum of the local coboundaries, namely,

\begin{theorem}[informal]
Every sufficiently small cochain $S:V_r\to\Ftwo$ in a direction-wise lossless cubical complex satisfying the imbalance assumptions admits cochains $L_t$, each supported in $\Lk(t)$, such that
\[
\bigcup_t\supp(L_t)=\supp(S),
\qquad
\norm{\delta S}\ge\sum_t\norm{\delta L_t}-C_{\eps,k}\norm{S},
\]
where $C_{\eps,k}\to 0$ as the loss parameters vanish.
\end{theorem}
For the definition of the norm, see \cref{eq:norm}. For a precise statement, see \Cref{ltg}.

Consequently, proving robust level expansion reduces to establishing expansion for the local structures $L_t$. Thus the proof naturally separates into two independent parts: a dimension-independent global-to-local decomposition and a purely local expansion theorem.

Let us now try to isolate the local components. We start with a cochain 
$S:V_r\rightarrow\Ftwo$. A first attempt would be to replace any two (or more) vertices in the support of $S$ that belong to the same link by that link, and then add each of the remaining isolated vertices in the support of $S$ as a link of its own (of size $1$). However, this decomposition is too fine to allow a global-to-local reduction. 

Instead, following \cite{FK24}, we use only those links which have a large intersection with $\supp{S}$.\footnote{The precise definition uses the direction-dependent factors encoded by the norm, which we suppress in the introduction.} We call vertices with such links \emph{heavy}. More accurately, the heavy vertices are organized recursively across the levels of the complex.
The set of heavy vertices on level $r$ is exactly  $H_r=\supp{S}$.
Then, starting from the support of the input cochain we go down the levels, repeatedly declaring a vertex to be heavy whenever a sufficiently large fraction of its neighbors in the layer above it are already heavy. This process terminates at a collection of \emph{terminal} vertices.
These terminal vertices, together with the hierarchy, give rise to the local cochains. Each terminal vertex $t$ gives a local cochain $L_t$, supported on all the vertices in $V_r$ that are above it in the heavy hierarchy.

The main difficulty is that these local cochains are not disjoint, so a single element can contribute to several local chains $L_t$, which may cause overcounting. 
Furthermore, the coboundaries of $L_t$ for different $t$'s may overlap, leading to cancellations over $\Ftwo$ in the next level $V_{r+1}$ that do not
occur when we count the size of the coboundary of each $L_t$ separately. The heart of our global-to-local reduction is to show that these overlaps are negligible.

The proof shows that the support of $\delta L_t$ usually does not intersect with the support of $\delta L_{t'}$ of a different link. To show this, we track the number of exceptions to this rule. When they do intersect on a vertex $v$, we charge this collision to the vertex $v$, and say that it has \emph{excess}.

% Vertices always inherit the excess of their heavy children. In \cref{lem:recursive-excess} we show that there are only two mechanisms by which a vertex can acquire additional excess. The first comes from collisions in a directional expander graph, and the other comes from a square completion to a non-heavy vertex. See \cref{fig:two-dimensional-square-complex} for an illustration of the two mechanisms.

To draw insight regarding the amount of excess from the geometry of the cubical complex we generalize and define excess for every heavy vertex as the number of superfluous terminals (more than $1$) it connects to. Then we start bounding the excess of vertices starting at the terminals and climbing up the layers. A vertex certainly connects to several terminals if its children do. The only other way for a vertex to have excess is by different children of it connecting to different terminals. In \cref{lem:recursive-excess} we show there are only two mechanisms for these uncooperative children. The first comes from failure of expansion in a directional expander graphs, and the other comes from a square completion of two children in different blocks to a non-heavy vertex. We bound each separately: the direction graphs are lossless expanders; and by the definition of heavy, a non-heavy vertex cannot close many squares. See \cref{fig:two-dimensional-square-complex} for an illustration of the two mechanisms.

As these two mechanisms are one- and two-dimensional in nature, they were already addressed, in some form, in \cite{HL22a,HL22b}\footnote{We stress that the proof of small-set coboundary expansion by \cite{HL22a,HL22b} differs from the one we give here, \emph{already in dimension 2}. In their work, the second mechanism bounds \emph{all} squares (light or not), while we apply a light/heavy distinction, and explicitly decompose the cochain. The distinction is required to handle the richer structure of cancellations over $\mathbb{F}_2$ that may occur in the coboundary of high dimensional cochains. }. Direction-wise lossless expansion completely controls the first mechanism, whereas for the second mechanism an additional imbalance condition is required. In \cref{sec:same-type,sec:light:squares,sec:mass-prop}, we show that requiring both expansion and imbalance implies that the total overlap between distinct local structures is small.

Crucially, local minimality \emph{plays no role} in the global-to-local reduction. It enters only in the local expansion theorem, where its purpose is to eliminate the local cohomological obstructions to expansion. Thus links and local minimality arise together: once links are identified as the local structures, local minimality becomes the natural notion of local irreducibility. Accordingly, we define a cochain to be locally minimal if no coboundary operation supported inside a single link decreases its norm. This is precisely the hypothesis under which the local expansion theorem is proved. 
In a sense, this approach mirrors the conceptual picture developed in the previous subsections. There, expansion and homology emerged as the two organizing principles of the theory. Both are inherently global notions. Here, they reappear in local form. Links identify the correct notion of locality, while local minimality is the corresponding notion of minimality on a single link and provides the corresponding local homological notion. A central insight of this work is that in \emph{direction-wise lossless} decorated cubical complexes, the expansion in the links is lifted automatically to the entire complex through the global-to-local reduction. In \cref{sec:local-picture}, we establish the required relaxed form of
coboundary expansion for the links. Combining these results yields the
following theorem:
\begin{theorem}[Robust small-set coboundary expansion at an arbitrary level]
\label{thm:arbitrary-rank-expansion}
Let $X$ be a $k$-dimensional lossless square complex, with parameters $\prn{D_i, D_i',c_i,\eps_i}$. If for every $s<t$, $\eps_tD_s' \leq \sum \eps_i$ then there exist  constants $c,a$ independent of $\abs{X(0)}$ such that $X$ is a robust $\prn{c, a-\eps}$ small-set coboundary expander, for $\eps = O_k\prn{\prn{\sum\eps_i}^{\Omega_k\prn{1}}}$
\end{theorem}

\subsection{Relation to concurrent work}

While this work was being completed, Gay and Jeronimo \cite{GJ26} gave the first construction of asymptotically good quantum locally testable codes. Their construction, like ours, is based on higher-dimensional cubical complexes, but the algebraic structures associated with the complexes and the mechanisms underlying the two approaches are rather different. Gay and Jeronimo work with a cosheaf complex whose cochain groups are built from the cubical faces. In contrast, we work with the level complex of a decorated cubical complex: all cochains are supported on vertices, with their degree determined by their level in the Boolean-cube decoration.  

The two approaches also make a different tradeoff between geometry and local codes. Gay and Jeronimo use arithmetic Ramanujan cubical complexes together with carefully chosen Reed--Solomon codes satisfying product-expansion properties. Thus, spectral expansion of the geometry is supplemented by strong expansion properties of the local codes. In our framework there are no nontrivial local codes: the burden is instead placed entirely on the geometry. We show that sufficiently strong lossless expansion of the one-dimensional directional graphs implies the higher-dimensional small-set level coboundary expansion needed for quantum local testability. Thus \cite{GJ26} establishes the existence of asymptotically good qLTCs, while the present work isolates a different, purely combinatorial route to such codes.

There is a useful historical parallel with the development of expander codes. Sipser and Spielman first identified lossless expansion as a clean combinatorial property through which sparse local constraints can enforce strong global coding properties, even though the explicit constructions available at the time were spectral expanders and their coding results were therefore also developed under spectral hypotheses. Subsequent work on lossless expanders showed that this stronger combinatorial notion is important in its own right, and eventually led to explicit constructions with increasingly strong parameters. We believe that a similar phenomenon may occur here. Lossless graph expansion has repeatedly proved useful beyond what can be obtained from spectral expansion alone, and our results suggest that its higher-dimensional cubical analogue is likewise a natural object to study. We therefore believe that suitably imbalanced cubical complexes with two-sided lossless directional expansion exist. Just as lossless expanders became fundamental objects independently of their original coding-theoretic motivation, we expect lossless cubical complexes to have applications beyond the particular qLTC construction considered here. We therefore view their construction as a natural higher-dimensional expansion problem in its own right.

%%%%%%%%%%%%%

\paragraph{Acknowledgements}

This work grew out of a summer reading group on qLTCs. We thank Aseel Sakhnini, Sergey Komech, and Yali Aharon Holzman for many helpful discussions. We also thank Shai Evra and Alex Lubotzky for illuminating conversations about constructions of lossless cubical complexes.

\paragraph{Use of AI}

We used ChatGPT 5.4-6 and Codex to assist with drafting and editing the manuscript, refining the exposition, and checking mathematical calculations. A preliminary version of \Cref{lem:recursive-excess} was obtained in a human-guided conversation with GPT-5.5, asking to extend our low-dimensional recursion using the same/light mechanism we proposed. All main ideas of this work are human-generated.
All AI-assisted material was reviewed by the authors, who take full responsibility for the content of this work. 

%% file: figures/left-right-cayley.tex
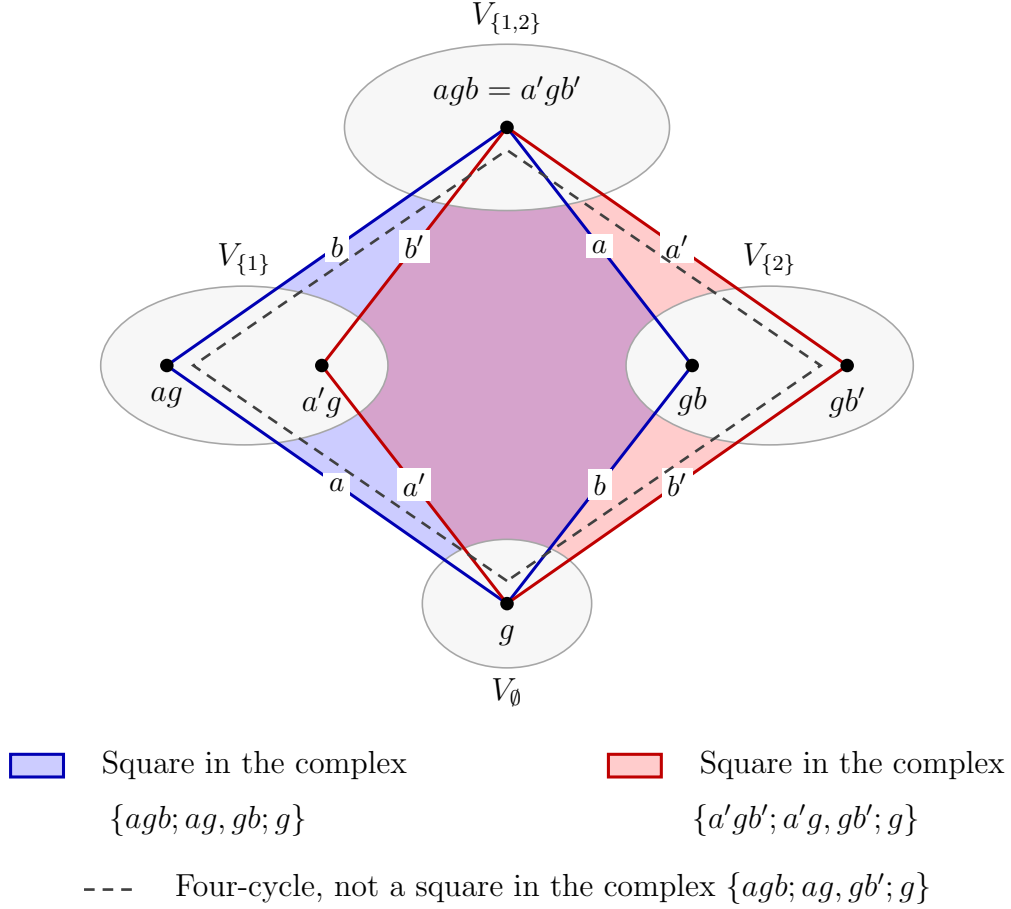
\begin{figure}[!ht]
\centering
\begin{tikzpicture}[
  x=1cm,y=1cm,font=\normalsize,
  block/.style={draw=black!35,fill=black!3,line width=.6pt},
  vertex/.style={circle,fill=black,inner sep=1.8pt},
  first square/.style={draw=blue!70!black,line width=1.15pt},
  second square/.style={draw=red!75!black,line width=1.15pt},
  false square/.style={draw=black!75,line width=1pt,
    dash pattern=on 4pt off 3pt},
  edge label/.style={fill=white,inner sep=2pt,font=\small}
]
  \coordinate (g) at (0,0);
  \coordinate (agb) at (0,6.3);
  \coordinate (ag) at (-4.5,3.15);
  \coordinate (apg) at (-2.45,3.15);
  \coordinate (gb) at (2.45,3.15);
  \coordinate (gbp) at (4.5,3.15);

  % The true square faces are 80 percent transparent (20 percent opaque).
  \fill[blue,fill opacity=.2] (g) -- (ag) -- (agb) -- (gb) -- cycle;
  \fill[red,fill opacity=.2] (g) -- (apg) -- (agb) -- (gbp) -- cycle;

  % Four disjoint copies of G; only the relevant vertices are displayed.
  \draw[block] (g) ellipse [x radius=1.12,y radius=.85];
  \draw[block] (agb) ellipse [x radius=2.15,y radius=1.1];
  \draw[block] (-3.475,3.15) ellipse [x radius=1.9,y radius=1.05];
  \draw[block] (3.475,3.15) ellipse [x radius=1.9,y radius=1.05];
  \node at (0,-1.17) {$V_\emptyset$};
  \node at (0,7.75) {$V_{\{1,2\}}$};
  \node at (-3.475,4.55) {$V_{\{1\}}$};
  \node at (3.475,4.55) {$V_{\{2\}}$};

  % Two genuine square boundaries, with the generator on each edge.
  \draw[first square] (g) -- node[edge label,pos=.5] {$a$} (ag)
    -- node[edge label,pos=.5] {$b$} (agb)
    -- node[edge label,pos=.5] {$a$} (gb)
    -- node[edge label,pos=.5] {$b$} cycle;
  \draw[second square] (g) -- node[edge label,pos=.5] {$a'$} (apg)
    -- node[edge label,pos=.5] {$b'$} (agb)
    -- node[edge label,pos=.5] {$a'$} (gbp)
    -- node[edge label,pos=.5] {$b'$} cycle;

  % An inset contour follows the mixed cycle; it denotes no new edges.
  \draw[false square] (0,.3) -- (-4.15,3.15)
    -- (0,6) -- (4.15,3.15) -- cycle;

  \foreach \v in {g,agb,ag,apg,gb,gbp}
    \node[vertex] at (\v) {};
  \node[anchor=north,inner sep=8pt] at (g) {$g$};
  \node[anchor=south,inner sep=8pt] at (agb) {$agb=a'gb'$};
  \node[anchor=north,inner sep=8pt] at (ag) {$ag$};
  \node[anchor=north,inner sep=8pt] at (apg) {$a'g$};
  \node[anchor=north,inner sep=8pt] at (gb) {$gb$};
  \node[anchor=north,inner sep=8pt] at (gbp) {$gb'$};

\end{tikzpicture}
\par\vspace{.35cm}
\begin{tabular}{@{}c@{\hspace{1cm}}c@{}}
  \makebox[.4\linewidth]{%
    \tikz[baseline=-.5ex]\draw[draw=blue!70!black,line width=1.15pt,
      fill=blue,fill opacity=.2] (0,-.15) rectangle (.7,.15);
    \hspace{.35cm}Square in the complex}
  &
  \makebox[.4\linewidth]{%
    \tikz[baseline=-.5ex]\draw[draw=red!75!black,line width=1.15pt,
      fill=red,fill opacity=.2] (0,-.15) rectangle (.7,.15);
    \hspace{.35cm}Square in the complex}
  \\[.22cm]
  $\set{agb; ag,gb; g}$ & $\set{a'gb'; a'g,gb'; g}$
  \\[.4cm]
  \multicolumn{2}{c}{%
    \tikz[baseline=-.5ex]\draw[draw=black!75,line width=1pt,
      dash pattern=on 4pt off 3pt] (0,0) -- (.7,0);
    \hspace{.35cm}Four-cycle, not a square in the complex $\set{agb; ag,gb'; g}$}
\end{tabular}
\caption{A left-right Cayley complex satisfying the weak completion property but not the strong one. The blue and red squares share the opposite vertices $g$ and $agb=a'gb'$. The black dashed contour follows a four-cycle which is not a square in the complex.}
\label{fig:left-right-weak-strong}
\end{figure}

%% file: Cubical-test.tex
\section{Cubical complexes}\label{sec:cubical}

We define cubical complexes alongside multipartite simplicial complexes, highlighting their similarities and connections. Along the way, we develop the notation and gradually build toward the statements of our main results.

Throughout, all complexes and vertex blocks are finite and nonempty. We write $[k]=\{1,\ldots,k\}$ and take all chains and cochains over $\Ftwo$.

\subsection{Simplicial and cubical complexes}

\begin{definition}[Simplicial complex]
A simplicial complex $Y$ consists of a vertex set $Y(0)$ and a downward-closed collection of nonempty subsets of $Y(0)$. For $i\geq 0$, let
\[
Y(i)=\{\sigma\in Y:|\sigma|=i+1\}
\]
be the set of $i$-dimensional faces.\footnote{We adjoin the empty face, of dimension $-1$, and set $Y(-1)=\{\varnothing\}$.} In particular, $Y(1)$ is the edge set of the underlying graph $(Y(0),Y(1))$. For $i\geq 0$, let
\[
\operatorname{Cl}_Y(i)=
\left\{\sigma\in\binom{Y(0)}{i+1}:\binom{\sigma}{2}\subseteq Y(1)\right\}.
\]
Thus $\operatorname{Cl}_Y(i)$ is the collection of $(i+1)$-vertex cliques in this graph, and $Y(i)\subseteq\operatorname{Cl}_Y(i)$.
\end{definition}

Equivalently, a simplicial complex is specified by its underlying graph and a collection of declared cliques from $\set{\operatorname{Cl}_Y(i)}_{i}$, closed under taking subsets. We now define multipartite complexes. To draw the analogy with decorated cubical complexes appearing later, we call them decorated simplicial complexes.

\begin{definition}[Decorated simplicial complex]
A simplicial complex $Y$ is $k$-decorated if $Y(0)=V_1\sqcup\cdots\sqcup V_k$ and every edge has endpoints in distinct blocks. Consequently, every face contains at most one vertex from each block.
\end{definition}

A decorated simplicial complex is modeled by the $(k-1)$-simplex: the decoration map $Y(0)\to[k]$, which sends every vertex in $V_j$ to $j$, constrains the $i$-faces in $Y$ (sets of size $i+1$) only to those whose image under the decoration map is an $i$-face. In that sense, the decoration gives a global model dictating which faces are allowed.

For cubical complexes, we use the $k$-dimensional Boolean cube $Q_k$ as our global model. Its vertices are the subsets of $[k]$, where $I$ and $J$ are adjacent precisely when they differ by one element $|I\mathbin{\triangle}J|=1$.

\begin{definition}[Cube-decorated graph]
A graph $G$ is $k$-cube-decorated if its vertex set is
partitioned into $2^k$ blocks $V_I$, indexed by
$I\subseteq[k]$, and every edge joins blocks $V_I,V_J$
with $\{I,J\}\in E(Q_k)$.
\end{definition}

For $v\in V_I$, write $\type(v)=I$.
An embedded $i$-cube in a $k$-cube-decorated graph $G$
is a set $Q\subseteq V(G)$ satisfying the following
conditions: (1) the vertices of $Q$ have distinct
types and $\{\type(v):v\in Q\}$ is the vertex set
of an $i$-face of $Q_k$; and (2) any two vertices of
$Q$ with adjacent types in $Q_k$ are adjacent in $G$.
We write $\operatorname{Cu}_G(i)$ for the set of all
embedded $i$-cubes in $G$.
An embedded cube $Q'\subseteq Q$ is called a subcube of $Q$.

\begin{definition}[Decorated cube complex]
A decorated $k$-cube complex $X$ is a $k$-cube-decorated
graph $G=(X(0),X(1))$ with sets
$X(i)\subseteq\operatorname{Cu}_G(i)$, $0\leq i\leq k$,
closed under subcubes, meaning that every $j$-subcube of
a cube in $X(i)$ belongs to $X(j)$.
\end{definition}

As with the cube $Q_k$, we regard $G$ as edge-colored by $[k]$: an edge between
$V_I$ and $V_J$ has color, or direction, $i$ when
$I\mathbin{\triangle}J=\{i\}$. An edge set is \emph{rainbow} if its directions are
pairwise distinct.

We now describe completion properties, which are crucial to this work.
A completion property is an additional assumption guaranteeing that certain
configurations extend to a unique cube in the complex, of the corresponding dimension.

\begin{definition}[Completion properties]
Let $1\leq r\leq k$.
\begin{enumerate}
\item $X$ has \emph{$r$-joint completion} if, for every $1\leq t\leq r$, every vertex together with $t$ incident edges in \emph{distinct} directions is contained in exactly one member of $X(t)$.
\item $X$ has \emph{$r$-path completion} if, for every $1\leq t\leq r$, every path of $t$ edges in \emph{distinct} directions is contained in exactly one member of $X(t)$.
\item $X$ has \emph{$r$-rainbow-tree completion} if, for every $1\leq t\leq r$, every \emph{rainbow} tree in $X(1)$ with $t$ edges is contained in exactly one member of $X(t)$.
\end{enumerate}
In each case, uniqueness is among the cubes in $X(t)$, and not with respect to $\operatorname{Cu}_G(t)$.
\end{definition}

\begin{lemma}\label{lem:completion-equivalence}
For every decorated $k$-cube complex and $1\leq r\leq k$, the completion properties are equivalent.
\[
\begin{aligned}
r\text{-joint completion}
&\quad\Longleftrightarrow\quad r\text{-rainbow-tree completion}
&\quad\Longleftrightarrow\quad r\text{-path completion}.
\end{aligned}
\]
\end{lemma}

\begin{proof}
We prove that the three properties imply one another cyclically. Both nontrivial implications use induction on the number $t\leq r$ of edges, with $t=1$ immediate.

\smallskip
\noindent\textbf{Joint $\Rightarrow$ rainbow.}
Let $T$ be a rainbow tree with $t\geq2$ edges, and delete a leaf $u$ adjacent to $v$. By induction, $T\setminus\set{u}$ has a unique completion $Q'\in X(t-1)$. Complete the $t$-joint formed by $uv$ and the edges of $Q'$ at $v$ to $Q\in X(t)$. Its face through $v$ in the directions of $Q'$ equals $Q'$ by joint completion, so $T\subseteq Q$. Any member of $X(t)$ containing $T$ contains $Q'$ by induction, hence this joint, and therefore equals $Q$.

\smallskip
\noindent\textbf{Rainbow $\Rightarrow$ path.}
Every rainbow path is a rainbow tree.

\smallskip
\noindent\textbf{Path $\Rightarrow$ joint.}
For a $t$-joint centered at $v$, remove an edge $e$ and complete the remaining joint to $Q'\in X(t-1)$ by induction. Append $e$ to a rainbow path of length $t-1$ in $Q'$ ending at $v$, and complete it to $Q\in X(t)$. Its face through $v$ in the directions of $Q'$ equals $Q'$ by path completion, so $Q$ contains the joint. Any member of $X(t)$ containing the joint contains $Q'$ by induction, hence this path, and therefore equals $Q$.
\end{proof}

\begin{definition}[$(k,r)$-cubical complex]
A decorated $(k,r)$-cubical complex is a decorated $k$-cube complex satisfying any of the equivalent $r$-completion properties. We subsequently omit the word decorated. A $(k,2)$-cubical complex is called a \emph{$k$-dimensional square complex}.
\end{definition}

\paragraph{The left--right Cayley complex.}
A standard example of a $(2,2)$-cubical complex is obtained from a finite group $\Gamma$ and symmetric subsets $A,B\subseteq \Gamma$, denoted $X_{A,\Gamma,B}$. For $I\subseteq\{1,2\}$, put $V_I=\Gamma\times\{I\}$ and write $g_I=(g,I)$. Direction one edges correspond to left multiplication by elements of $A$; direction two edges correspond to right multiplication by elements of $B$. The squares in $X_{A,\Gamma,B}(2)$ are precisely the squares
\[
Q(g,a,b)=
\left\{g_\varnothing,(ag)_{\{1\}},(gb)_{\{2\}},(agb)_{\{1,2\}}\right\},
\qquad g\in \Gamma,\ a\in A,\ b\in B.
\]
It satisfies the completion axiom. Indeed, the joint at $g_\varnothing$ with neighbors $(ag)_{\{1\}}$ and $(gb)_{\{2\}}$ has the unique completion $Q(g,a,b)$; joints at the other vertex types are treated similarly. See \cref{fig:left-right-weak-strong} for an illustration.

\begin{remark}
An embedded cube $Q\in \operatorname{Cu}_X(i)$ need not be in $X(i)$. For example, when $A=B=\Gamma$, $\operatorname{Cu}_{X_{A,\Gamma,B}}(2)$ contains $|\Gamma|^4$ four-cycles with one vertex in each block, whereas only the $|\Gamma|^3$ squares $Q(g,a,b)$ are in $X_{A,\Gamma,B}(2)$. Thus our definition does not impose a ``flag'' condition requiring every embedded cube to be declared.
\end{remark}

\begin{remark}
One may also impose stronger uniqueness properties: whenever $x\in V_I$ and $y\in V_J$ are joined by a rainbow path of length $d=|I\mathbin{\triangle}J|$, they belong to a unique declared $d$-cube. This requires uniqueness from the endpoints alone. The LPS-based cubical complexes constructed in \cite[Lemma~3.3 and Theorem~3.5]{HLMRZ25} satisfy this condition, but left--right Cayley complexes need not, as depicted in \cref{fig:left-right-weak-strong}.
\end{remark}

\subsection{The level chain complex}

Let $X$ be a $(k,2)$-cubical complex. Our chains and cochains are functions on the vertices of $X$, rather than on its cubical faces. Taking $\varnothing$ as the origin of the model cube, we write $V_s=\bigsqcup_{|I|=s}V_I$ for the $s$-th level, $0\leq s\leq k$. Unless specified otherwise, subscripts on vertex sets and functions indicate restriction to the corresponding type or level. Thus, for \(A\subseteq X(0)\), we write \(A_I=A\cap V_I\) and \(A_s=A\cap V_s\); similarly, \(f_I=f|_{V_I}\) and \(f_s=f|_{V_s}\). When viewing a scalar-valued function on a subset of vertices as an ambient function, we extend it by zero.

For $i\notin I$, orient every edge between $V_I$ and $V_{I\cup\{i\}}$ toward $V_{I\cup\{i\}}$. Write $x\lessdot y$ when an edge is oriented from $x$ to $y$, and $x\leq y$ when $x=y$ or there is an upward path from $x$ to $y$.
If $b,b'$ are neighbors of $a$ in distinct directions, denote the fourth vertex of their unique square completion by $b\meet[a]b'$. In particular, if $b,b'\lessdot a$, then $b\meet[a]b'\lessdot b,b'$. For readability, we write \(\{a;b,b';c\}\) for a square with \(a\in V_i\), \(b,b'\in V_{i-1}\), and \(c\in V_{i-2}\).

For $0\leq s\leq k$, define $C_s(X)=C^s(X)=\Ftwo^{V_s}$. The boundary map $\partial_s:C_s(X)\to C_{s-1}(X)$ and the coboundary map $\delta_s:C^s(X)\to C^{s+1}(X)$ are given by
\[
(\partial_s f)(u)=\sum_{u\lessdot v}f(v),
\qquad
(\delta_s f)(w)=\sum_{v\lessdot w}f(v),
\]
for $1\leq s\leq k$ and $0\leq s<k$, respectively. All sums are over $\Ftwo$, and $\partial_s=\delta_{s-1}^{T}$. The following lemma shows that square completion gives the identities $\partial^2=\delta^2=0$. The following lemma and its proof are from \cite{HL22a,HL22b}.

\begin{lemma}\label{chain:complex}
For every $(k,2)$-cubical complex $X$ and $1\leq s\leq k-1$,
\[
\partial_s\partial_{s+1}=0,
\qquad
\delta_s\delta_{s-1}=0.
\]
\end{lemma}

\begin{proof}
Fix $x\in V_{s+1}$ and $z\in V_{s-1}$. The $(z,x)$-entry of $\partial_s\partial_{s+1}$ counts, modulo two, the vertices $y\in V_s$ satisfying $z\lessdot y\lessdot x$. For each such $y$, the unique square completing along this path has another middle vertex
\[
y'=z\meet[y]x.
\]
Observe that completing $z\lessdot y'\lessdot x$ gives the same square, i.e. $z \meet[y'] x = y$, so $y\mapsto y'$ is a fixed-point-free involution, since $\type(y) \neq \type(y')$ . Therefore the number of vertices $y\in V_s$ satisfying $z\lessdot y\lessdot x$ is even, proving $\partial_s\partial_{s+1}=0$. Taking transposes yields $\delta_s\delta_{s-1}=0$.
\end{proof}

For $I\subseteq[k]$ and $i\notin I$, let $X_{I,i}$ be the bipartite graph between $V_I$ and $V_{I\cup\{i\}}$. We assume throughout that $X_{I,i}$ is $(D_i,D_i')$-biregular, with $D_i,D_i'>0$: vertices in $V_I$ have degree $D_i$, and vertices in $V_{I\cup\{i\}}$ have degree $D_i'$. In particular, throughout, degrees depend only on the direction $i$, but not on $I$.

\paragraph{Norms.}\label{sec:norm}
For $I\subseteq[k]$, put $D_I=\prod_{i\in I}D_i$ and $D_I'=\prod_{i\in I}D_i'$, with $D_\varnothing=D_\varnothing'=1$. Assign each vertex $v\in V_I$ weight $\norm{v}_X=\norm{I}_X=D_I^{-1}$. For $f:X(0)\to\mathbb R$, define
\begin{equation}\label{eq:norm}
    \norm{f}_X = \sum_{I\subseteq[k]}\sum_{v\in V_I}\frac{|f(v)|}{D_I}.
\end{equation}
When the ambient complex is clear from context we omit the subscript. For an $\Ftwo$-valued chain or cochain $f$, its norm is the total weight of its support: $\norm{f}=\sum_{v\in\supp f}\norm{v}$. We identify a set with its indicator function, writing $\norm{H}=\norm{\mathbf 1_H}$. With addition over $\Ftwo$, we have $\norm{f+g}\leq\norm{f}+\norm{g}$.
This norm extends to $k$-partite simplicial complexes. A face of type $I$ is assigned the norm $D_I^{-1}$, and the norm of a cochain is the total weight of its support. Throughout, we use the normalized distance. For a cochain $f_r\in C^r(X)$ and $C\subseteq C^r(X)$, write
\[
    \dist_X(f_r, C)
    \defeq
    \min_{c\in C}
    \norm{f_r+c}, 
    \qquad
    \dist_X(f_r) = \dist\prn{f_r, \im (\delta^{r-1})}
\]
Here too, the subscript is omitted when the ambient complex is clear from context.

\begin{remark}\label{rem:hamming-vs-norm}
For every $f:X(0)\to\Ftwo$,
\[
\frac{|\supp f|}{D_{[k]}}
\leq \norm{f}
\leq |\supp f|.
\]
Thus the weighted norm and Hamming weight are equivalent up to a factor depending only on the degrees.
\end{remark}

For $v\in V_I$ and $i\notin I$, the total weight of its neighbors in direction $i$ is $D_i/D_{I\cup\{i\}}=\norm{v}$. Summing over the upward directions gives
\begin{equation}\label{eq:norm-propagation}
\norm{\delta_{|I|}\mathbf 1_{\{v\}}}
=\sum_{v\lessdot w}\norm{w}
=\sum_{i\notin I}\frac{D_i}{D_{I\cup\{i\}}}
=(k-|I|)\norm{v}.
\end{equation}

Finally, choosing $[k]$ as the origin reverses every edge orientation. The resulting \emph{reversed complex} $X^{\mathrm{rev}}$ has the same vertices, edges, and declared cubes, with $V_I^{\mathrm{rev}}=V_{[k]\setminus I}$. The following lemma relates its coboundary maps to the boundary maps of $X$.

\begin{lemma}\label{lem:reversal}
For $0\leq s\leq k$, the identity on functions identifies $C_s(X)$ with $C^{k-s}(X^{\mathrm{rev}})$. Under this identification,
\[
\partial_s^X=\delta_{k-s}^{X^{\mathrm{rev}}}
\qquad(1\leq s\leq k).
\]
Moreover, $D_i^{\mathrm{rev}}=D_i'$ and $(D_i')^{\mathrm{rev}}=D_i$.
\end{lemma}

\begin{proof}
By definition, $V_s(X)=V_{k-s}(X^{\mathrm{rev}})$, and $x\lessdot_X y$ if and only if $y\lessdot_{X^{\mathrm{rev}}}x$. Thus the corresponding boundary and coboundary maps sum over the same neighbors. Reversal also exchanges the two sides of every directional graph, giving the degree identities.
\end{proof}

This identification need not preserve norms. For $f\in C_s(X)$, the reversed norm is
\[
\norm{f}_{X^{\mathrm{rev}}}
=\sum_{\substack{I\subseteq[k]\\|I|=s}}\sum_{v\in V_I}
\frac{|f(v)|}{D'_{[k]\setminus I}}.
\]

\subsection{Links}

\begin{definition}[Link]\label{def:augmented-link}
Let $v\in V_I$. For $J\subseteq[k]\setminus I$, define
\[
\Lk_J(v)=\{x\in V_{I\cup J}:v\leq x\}.
\]
The link of $v$ is $\Lk(v)=\{x\in X(0):v\leq x\}$, with levels
\[
\Lk_s(v)=\bigsqcup_{\substack{J\subseteq[k]\setminus I\\|J|=s+1}}\Lk_J(v),
\qquad -1\leq s\leq k-|I|-1.
\]
In particular, $\Lk_{-1}(v)=\Lk_{\emptyset}(v)=\set{v}$.
\end{definition}

The chain and cochain groups of the link are $C_s(\Lk(v))=C^s(\Lk(v))=\Ftwo^{\Lk_s(v)}$. Their boundary and coboundary maps use the same incidence formulas as above, with sums restricted to vertices in the link. Thus $\delta_{-1}$ sends a value at $v$ to the constant cochain with that value on $\Lk_0(v)$.

For $h\in C^s(\Lk(v))$, let $\iota_vh\in C^{|I|+s+1}(X)$ be its extension by zero. The link norm is
\begin{equation}\label{eq:link-norm}
\norm{h}_v=
\sum_{\substack{J\subseteq[k]\setminus I\\|J|=s+1}}\sum_{x\in\Lk_J(v)}\frac{|h(x)|}{D_J}.
\end{equation}
Since $D_{I\cup J}=D_ID_J$, we have $\norm{\iota_vh}_X=D_I^{-1}\norm{h}_v$.

\begin{lemma}\label{lem:link-chain-complex}
If $X$ is a $(k,2)$-cubical complex, then the restricted link maps satisfy $\partial^2=0$ and $\delta^2=0$.
\end{lemma}

\begin{proof}
The involution of \cref{chain:complex} preserves the link.
\end{proof}

\subsection{Small-set coboundary expansion}

\begin{definition}[Minimality]\label{def:minimality}
For $1\leq r\leq k$, a cochain $f\in C^r(X)$ is \emph{minimal} if
\[
\norm{f}\leq\norm{f+\delta_{r-1}h}
\qquad\text{for every }h\in C^{r-1}(X).
\]
It is \emph{locally minimal} if, for every $0\leq s\leq r-1$, every $v\in V_s$, and every $h\in C^{r-s-2}(\Lk(v))$,
\[
\norm{f}\leq\norm{f+\delta_{r-1}(\iota_vh)}.
\]
All cochains in \(C^0(X)\) are minimal.
\end{definition}

Every minimal cochain is locally minimal, since each $\iota_vh$ in the definition belongs to $C^{r-1}(X)$.

\begin{definition}[Small-set coboundary expansion]\label{def:small-set-coboundary-expansion}
For $0\leq r\leq k-1$, a cubical complex $X$ is an $(\alpha,\beta)$-small-set coboundary expander at level $r$ if every locally minimal $f\in C^r(X)$ satisfying $\norm{f}\leq\alpha$ also satisfies $\norm{\delta_rf}\geq\beta\norm{f}$.
\end{definition}

\subsection{Lossless expansion}
For a graph $H$, write $\Gamma(v)$ for the neighborhood of $v$ and $\Gamma(A)=\bigcup_{v\in A}\Gamma(v)$. For $v\in V_I$ and $i\in[k]$, let $\Gamma_i(v)$ denote the neighbors of $v$ in direction $i$. Put $\Gamma_i(A)=\bigcup_{v\in A}\Gamma_i(v)$. Thus $\Gamma_i(v)\subseteq V_{I\mathbin{\triangle}\{i\}}$.

\begin{definition}[One-sided lossless bigraph]
Let $H=(L,R,E)$ be a $(D,D')$-biregular bigraph. For $c,\varepsilon\in\prn{0,1}$, we say that $H$ is \emph{one-sided $(c,\varepsilon)$-lossless} from $L$ to $R$ if every $A\subseteq L$ with $|A|\leq c|L|$ satisfies
\[
|\Gamma(A)|\geq(1-\varepsilon)D|A|.
\]
\end{definition}

\begin{definition}[Lossless cubical complex]\label{def:lossless-cubical:complex}
A cubical complex $X$ is \emph{lossless} with parameters $(D_i,D_i',c_i,\varepsilon_i)_{i=1}^k$ if, for every $I\subseteq[k]$ and $i\notin I$, the directional graph $X_{I,i}$ is one-sided $(c_i,\varepsilon_i)$-lossless from $V_I$ to $V_{I\cup\{i\}}$. It is \emph{two-sided lossless} if $X^{\mathrm{rev}}$ is also lossless, with $D_i,D_i'$ interchanged and the same $c_i,\varepsilon_i$.
\end{definition}

For a bigraph $H=(L,R,E)$, let $e(A,B)$ denote the number of edges between $A\subseteq L$ and $B\subseteq R$. For a vertex $v$ and a vertex set $S$, put $\deg_S(v)=|\Gamma(v)\cap S|$. Denote $t_+=\max\{t,0\}$. 
\begin{lemma}\label{lem:lossless-estimates}
Let $H=(L,R,E)$ be a one-sided $(c,\varepsilon)$-lossless $(D,D')$-biregular bigraph.
\begin{enumerate}
\item\label{lossless:excess} \textnormal{\textbf{Excess bound.}} If $A\subseteq L$ satisfies $|A|\leq c|L|$, then
\[
\sum_{y\in R}(\deg_A(y)-1)_+
=D|A|-|\Gamma(A)|
\leq\varepsilon D|A|.
\]
\item\label{lossless:edge} \textnormal{\textbf{Lossless edge bound.}} If $A\subseteq L$ satisfies $|A|\leq c|L|$, then, for every $B\subseteq R$,
\[
e(A,B)\leq |B|+\varepsilon D|A|.
\]
\item\label{lossless:large:degree} \textnormal{\textbf{Large-degree bound.}} For $B\subseteq R$, suppose that at most $c|L|$ vertices $x\in L$ satisfy $\deg_B(x)>\varepsilon D$. Then
\[
\sum_{x\in L}\left(\frac{\deg_B(x)}D-\varepsilon\right)_+
\leq\frac{|B|}D.
\]
\end{enumerate}
\end{lemma}

\begin{proof}
\noindent\textbf{Excess bound (\ref*{lossless:excess}).}
Fix $A\subseteq L$ with $|A|\leq c|L|$. Every vertex in $\Gamma(A)$ has at least one neighbor in $A$, and vertices outside $\Gamma(A)$ have none. Consequently, by regularity and lossless expansion,
\[
\sum_{y\in R}(\deg_A(y)-1)_+
=\sum_{y\in\Gamma(A)}(\deg_A(y)-1)
=D|A|-|\Gamma(A)|
\leq\varepsilon D|A|.
\]

\smallskip
\noindent\textbf{Lossless edge bound (\ref*{lossless:edge}).}
Fix $A\subseteq L$ with $|A|\leq c|L|$ and $B\subseteq R$. At each vertex of $B\cap\Gamma(A)$, count one incident edge from $A$ separately from the remaining edges. This gives
\[
e(A,B)
=|B\cap\Gamma(A)|+\sum_{y\in B}(\deg_A(y)-1)_+
\leq |B|+\sum_{y\in R}(\deg_A(y)-1)_+
\leq |B|+\varepsilon D|A|.
\]
The last inequality follows from part~(\ref{lossless:excess}).

\smallskip
\noindent\textbf{Large-degree bound (\ref*{lossless:large:degree}).}
Fix $B\subseteq R$ as in the statement, and put $A=\{x\in L:\deg_B(x)>\varepsilon D\}$, so $|A|\leq c|L|$. The summand $(\deg_B(x)/D-\varepsilon)_+$ is positive precisely when $x\in A$. Therefore,
\[
\sum_{x\in L}\left(\frac{\deg_B(x)}D-\varepsilon\right)_+
=\sum_{x\in A}\left(\frac{\deg_B(x)}D-\varepsilon\right)
=\frac{e(A,B)}D-\varepsilon|A|.
\]
The bound $|A|\leq c|L|$ allows us to apply part~(\ref{lossless:edge}), obtaining
\[
e(A,B)\leq |B|+\varepsilon D|A|.
\]
Substituting this bound into the preceding identity gives the required upper bound $|B|/D$.
\end{proof}

%% file: Local.tex
\section{Expansion of the links}
\label{sec:local-picture}
In this section we show that links in lossless cubical complexes can be approximated by complete multipartite simplicial complexes. Recall that we define the norm on such complexes similarly to the norm on cubical complexes. If the parts have sizes $D_i$ for $i\in[d]$, then a face of type $I$ has weight $D_I^{-1}$, and the norm of a cochain is the sum of the weights of the faces in its support. Under this norm, complete $d$-partite complexes are known to expand minimal cochains at every level $-1\leq r\leq d-2$, with a constant factor $m_{d,r}$ depending on $d$ and $r$ but not on the part sizes:
\begin{equation}
    \label{eq:m}
m_{d,r}:=
\left(\binom d{r+1}\binom{d+r+1}{r+2}\right)^{-1}.
\end{equation}

\begin{theorem}[{\cite[Corollary~3.1]{LMM16}}]\label{thm:lmm-multipartite}
Let $\M=K(D_1,\ldots,D_d)$ be the complete $d$-partite complex with block sizes $D_1,\ldots,D_d$. For $-1\leq r\leq d-2$, every $f\in C^r(\M)$ satisfies
\[
\norm{\delta_rf}\geq m_{d,r}\dist(f).
\]
\end{theorem}

We prove that lossless square complexes inherit this expansion in
their links, up to an additive error depending on the loss. 
\begin{theorem}
\label{thm:link-expansion-local}
Let $X$ be a $k$-dimensional lossless square complex with parameters
$\prn{D_i,D_i',c_i,\eps_i}_{i=1}^k$.
Fix $I\subseteq[k]$ with $|I|<k$ and a vertex $v\in V_I$, and put
$d=k-|I|$. Then, for every $-1\leq r\leq d-2$ and every cochain
$f:\Lk_r(v)\to\Ftwo$,
\begin{equation}
    \label{eq:link-expansion}
    \norm{\delta_r f}_v
    \ge
    m_{d,r}
    \dist(f)
    -O_k\prn{\sum_{j=1}^k\eps_j}.
\end{equation}
\end{theorem}
To simplify notations, we first treat the link of a vertex $v\in V_\emptyset$ and reduce to this case in the end.

\subsection{Ordered paths and collisions}

Let $v\in V_\emptyset$ and let $I\subseteq[k]$, with $|I|=i$.  Consider the
set of monotone paths from $v$ to $\Lk_I(v)$,
\[
    \mathcal P
    \defeq
    \set{(x_0,\ldots,x_i):
        x_0=v,\ x_i\in\Lk_I(v),\
        x_{j-1}\lessdot x_j\text{ for every }j\in[i]}.
\]
\par\noindent
Every path in $\mathcal P$ uses each direction in $I$ exactly once, though
the order can vary.  Define
\begin{align*}
    \Ord &: \mathcal P\longrightarrow
        \set{\pi:[i]\twoheadrightarrow I},
    &\set{\Ord((x_0,\ldots,x_i))(j)}
        &\defeq\type(x_j)\setminus\type(x_{j-1}), \\
    \End &: \mathcal P\longrightarrow\Lk_I(v),
    &\End((x_0,\ldots,x_i))
        &\defeq x_i.
\end{align*}

That is, $\Ord$ records the order in which the directions appear along the path, and $\End$ stands for the other endpoint of the path.  For an ordering $\pi$ of $I$, let
\begin{align*}
    \mathcal P_\pi
    &\defeq
    \set{P\in\mathcal P:\Ord(P)=\pi},
    &
    \End_\pi
    &\defeq
    \End|_{\mathcal P_\pi}.
\end{align*}
At the $j$th step of a path in $\mathcal P_\pi$ there are
$D_{\pi(j)}$ choices, so
\begin{equation}
    \label{eq:number-ordered-paths}
    \abs{\mathcal P_\pi}=\prod_{j=1}^{i} D_{\pi(j)}=D_I.
\end{equation}

\begin{lemma}[Switching ordered paths]
\label{lem:ordered-path-switching}
For every $x\in\Lk_I(v)$, the fiber size
$\abs{\End_\pi^{-1}(x)}$ is independent of the ordering $\pi$.
\end{lemma}

\begin{figure}[H]
\centering
\begin{tikzpicture}[
    every node/.style={font=\small},
    rankone/.style={draw=blue!45!black,fill=blue!5,rounded corners=11pt,
        minimum width=2.6cm,minimum height=1.35cm},
    ranktwo/.style={draw=blue!45!black,fill=blue!7,rounded corners=13pt,
        minimum width=3.2cm,minimum height=1.65cm},
    rankthree/.style={draw=blue!45!black,fill=blue!10,rounded corners=16pt,
        minimum width=5.2cm,minimum height=1.9cm},
    cube edge/.style={gray!75,densely dotted,line width=.45pt},
    path vertex/.style={circle,fill=blue!70!black,inner sep=1.7pt},
    alternate vertex/.style={circle,fill=orange!85!black,inner sep=1.7pt},
    shared vertex/.style={circle,fill=blue!70!black,draw=orange!85!black,
        line width=.9pt,inner sep=1.9pt},
    other vertex/.style={circle,fill=gray!55,inner sep=1.3pt}
]
    \coordinate (v) at (0,0);

    % Filling the four boundary edges marks the square as a two-cell.
    \filldraw[
        fill=yellow!35,
        fill opacity=.42,
        draw=yellow!55!orange,
        draw opacity=.75,
        line width=.6pt
    ]
        (v)--(-6.25,1.7)--(-5.55,4.35)--(-.35,1.7)--cycle;

    \node[rankone] (X1) at (-6,2) {};
    \node[rankone] (X2) at (0,2) {};
    \node[rankone] (X3) at (6,2) {};
    \node[ranktwo] (X12) at (-6,4.8) {};
    \node[ranktwo] (X13) at (0,4.8) {};
    \node[ranktwo] (X23) at (6,4.8) {};
    \node[rankthree] (X123) at (0,7.8) {};

    % The dotted skeleton records inclusion of types.
    \foreach \x in {X1,X2,X3}
        \draw[cube edge] (v)--(\x.south);
    \foreach \x/\y in {X1/X12,X1/X13,X2/X12,X2/X23,X3/X13,X3/X23}
        \draw[cube edge] (\x.north)--(\y.south);
    \foreach \x in {X12,X13,X23}
        \draw[cube edge] (\x.north)--(X123.south);

    % The paths differ at level one and coincide from level two onward.
    \draw[blue!70!black,solid,line width=1.15pt]
        (v)--(-6.25,1.7)--(-5.55,4.35)--(-.8,7.25);
    \draw[orange!85!black,dashed,line width=1.15pt]
        (v)--(-.35,1.7)--(-5.55,4.35)--(-.8,7.25);

    \node[path vertex] at (-6.25,1.7) {};
    \node[alternate vertex] at (-.35,1.7) {};
    \node[shared vertex] at (-5.55,4.35) {};
    \node[shared vertex] at (-.8,7.25) {};
    \node[shared vertex] at (v) {};

    % Gray dots show the other vertices in each part.
    \foreach \p in {
        (-5.85,1.7),
        (-5.5,1.7),
        % (0,1.7),
        (.35,1.7),
        (5.7,1.7),
        (6,1.7),
        (6.3,1.7),
        (-6.65,4.35),
        (-6.15,4.35),
        (-6.4,4.68),
        (-5.9,4.7),
        (-.55,4.35),
        (0,4.32),
        (.55,4.35),
        (-.28,4.7),
        (.3,4.68),
        (5.45,4.35),
        (6,4.32),
        % (6.55,4.35),
        (5.72,4.7),
        (6.3,4.68),
        (-.3,7.22),
        (.25,7.25),
        (.8,7.22),
        (-1.1,7.58),
        (-.55,7.58),
        (0,7.58),
        (.55,7.58),
        (1.1,7.58)
        }
        \node[other vertex] at \p {};

    \node[below=2pt] at (v) {$x_0=v$};
    \node[blue!70!black,below left=1pt]
        at (-6.25,1.8) {$x_1$};
    \node[orange!85!black,below left=1pt]
        at (-.35,1.9) {$x_1'$};
    \node[right=5pt]
        at (-5.55,4.35) {$x_2$};
    \node[below right=1pt]
        at (-.8,7.25) {$x_3$};
    \node at (-5.6,2.25) {$\Lk_{\{1\}}(v)$};
    \node at (0,2.25) {$\Lk_{\{2\}}(v)$};
    \node at (5.6,2.25) {$\Lk_{\{3\}}(v)$};
    \node at (-6.35,5.1) {$\Lk_{\{1,2\}}(v)$};
    \node at (0,5.1) {$\Lk_{\{1,3\}}(v)$};
    \node at (6.35,5.1) {$\Lk_{\{2,3\}}(v)$};
    \node at (0,8.15) {$\Lk_{\{1,2,3\}}(v)$};
    \node[blue!70!black,above left=1pt] at (-3.1,.84) {$P$};
    \node[orange!85!black,right=4pt] at (-.18,.85) {$P'$};
    \node[
        fill=yellow!12,
        fill opacity=.85,
        text opacity=1,
        rounded corners=2pt,
        inner sep=1.5pt,
        font=\small,
        text=gray!65!black
    ] at (-3,2) {$\{x_2;\,x_1,x_1';\,x_0\}$};

\end{tikzpicture}
\caption{The link above $v$ in a three-dimensional cubical complex. The solid path
$P$ has order $(1,2,3)$, while the dashed path $P'$ has order $(2,1,3)$.
The pale shaded quadrilateral marks the square $2$-cell witnessing
the adjacent switch.}
\label{fig:ordered-path-switch}
\end{figure}
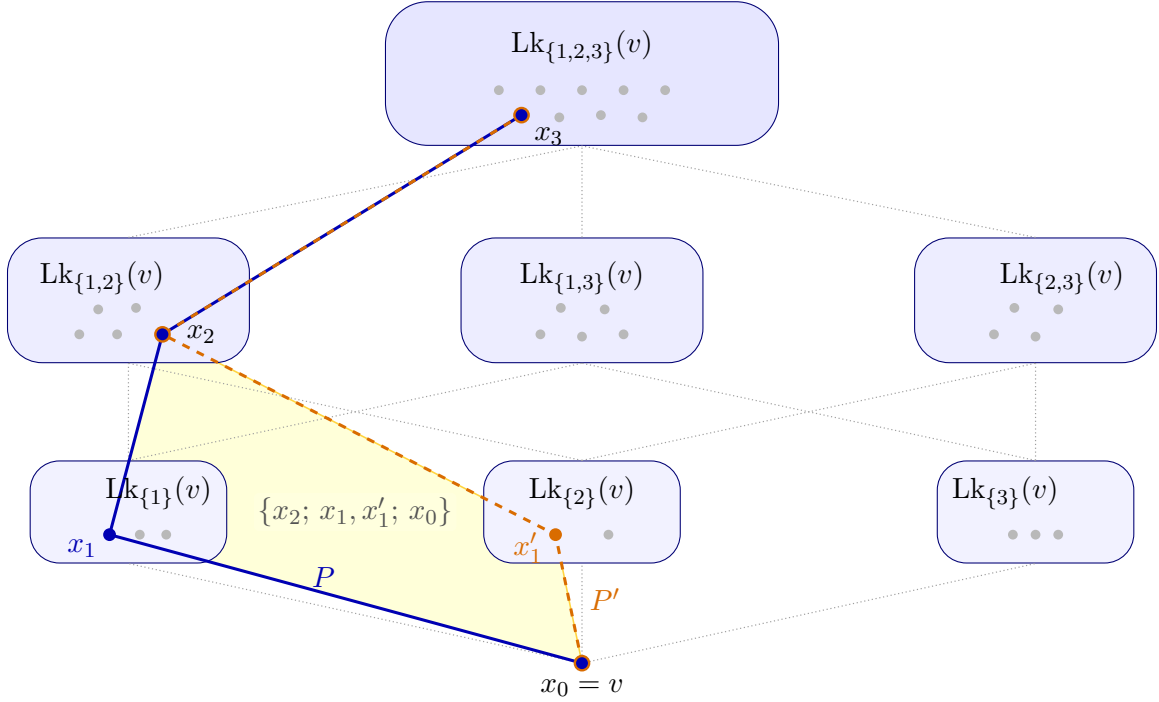

\begin{proof}
Given $P=(x_0,\ldots,x_i)\in\End_\pi^{-1}(x)$, fix $1\le j<i$.  The segment
$x_{j-1}\lessdot x_j\lessdot x_{j+1}$ is a two-edge joint at $x_j$.
The square-completion axiom supplies a unique square with other
intermediate vertex
$x_j'=x_{j-1}\meet[x_j]x_{j+1}$.  Since
$v\le x_{j-1}\lessdot x_j'$, this vertex also belongs to $\Lk(v)$.  Replacing
$x_j$ by $x_j'$ gives the path
\[
    (x_0,\ldots,x_{j-1},x_j',x_{j+1},\ldots,x_i),
\]
which has the same endpoint and swaps the $j$th and $(j+1)$st directions, as
illustrated in \cref{fig:ordered-path-switch}.  Repeating the switch recovers
the original path.  It therefore gives a bijection between
$\End_\pi^{-1}(x)$ and
$\End_{\pi\circ\sigma_{j,j+1}}^{-1}(x)$, where
$\sigma_{j,j+1}$ is the adjacent transposition.  Since adjacent
transpositions generate the symmetric group, all fibers have the same
size.
\end{proof}

\begin{corollary}
\label{cor:ordered-endpoint-surjective}
For every ordering $\pi$ of $I$, the endpoint map $\End_\pi$ is surjective.
\end{corollary}

\begin{proof}
Let $x\in\Lk_I(v)$.  Some monotone path $P$ ends at $x$; if
$\rho=\Ord(P)$, then $\abs{\End_\rho^{-1}(x)}\ge1$.  By
\cref{lem:ordered-path-switching}, the same holds for every ordering $\pi$.
\end{proof}

\begin{corollary}
\label{cor:link-ub}
For every $I\subseteq[k]$ and every ordering $\pi$ of $I$,
\[
    \norm{\Lk_I(v)}\le\frac{1}{D_I}\abs{\mathcal P_\pi}=1.
\]
\end{corollary}

Lossless expansion gives the reverse inequality up to the loss.  Fix an
ordering $\pi\prn{j}=s_j$ and put
$I_t\defeq\set{s_1,\ldots,s_t}$, with $I_0=\emptyset$.  By
\cref{cor:ordered-endpoint-surjective}, the endpoints of all length-$t$
paths with ordering $\pi|_{[t]}$ are precisely $\Lk_{I_t}(v)$, and hence
\[
    \Lk_{I_t}(v)
    =
    \Gamma_{s_t}\prn{\Lk_{I_{t-1}}(v)}.
\]
Successive lossless expansion gives
\begin{equation}
    \label{eq:link-part-lower-bound}
    \norm{\Lk_I(v)}
    \ge
    \prod_{j\in I}(1-\eps_j)
    \ge
    1-\sum_{j\in I}\eps_j.
\end{equation}

For each type $I$, fix a reference ordering and define, for
$x\in\Lk_I(v)$,
\[
    \mu_I(x)
    \defeq
    \abs{\End_\pi^{-1}(x)}-1,
    \qquad
    B_I
    \defeq
    \set{x\in\Lk_I(v):\mu_I(x)>0}.
\]
The function $\mu_I$ measures the endpoint excess and $B_I$ is the set of
bad, or colliding, endpoints.  By \eqref{eq:number-ordered-paths} and
\cref{cor:link-ub} \begin{equation}
    \label{eq:endpoint-excess}
    \begin{aligned}
    \norm{\mu_I}
    =
    \frac{1}{D_I}
    \sum_{x\in\Lk_I(v)}
        \prn{\abs{\End_\pi^{-1}(x)}-1} 
    =
    \frac{\abs{\mathcal P_\pi}-\abs{\Lk_I(v)}}{D_I}
    =
    1-\norm{\Lk_I(v)}
    \end{aligned}
\end{equation}
Since the indicator of $B_I$ is pointwise at most $\mu_I$, and using \eqref{eq:link-part-lower-bound}, we get
\begin{equation}
    \label{eq:actual-bad-type-bound}
    \norm{B_I}\le\norm{\mu_I}\le\sum_{j\in I}\eps_j.
\end{equation}

Put $B\defeq\bigcup_I B_I$.  The set $B$ is
upward closed.  Indeed, two distinct paths with the same ordering ending at
$x$ remain distinct after appending any fixed monotone path from $x$ to an
ancestor $y$.  By \cref{lem:ordered-path-switching}, a duplicated fiber for
one ordering is duplicated for every ordering, including the reference
ordering used to define the bad set containing $y$.  Thus $x\in B$ and
$x\le y$ imply $y\in B$.  Equivalently, defining the non-colliding set by $G\defeq\Lk(v)\setminus B$, we obtain that $G$ is
downward closed within $\Lk(v)$.
\subsection{Embedding the non-colliding part}

Let $\M$ be the complete $k$-partite simplicial complex whose $j$th part is
$\Lk_{\set{j}}(v)=\Gamma_j(v)$.

\begin{claim}
\label{clm:good-interval}
For every $x_I\in G$, the interval $[v,x_I]$\footnote{As usual
for a poset, $[v,x_I]=\set{y:v\le y\le x_I}$.}
satisfies the following:
\begin{enumerate}
\item[(1)] For every $J\subseteq I$, there is a unique vertex
$x_J\in[v,x_I]\cap\Lk_J(v)$.
\item[(2)] For every $J,K\subseteq I$, we have $x_J\le x_K$
if and only if $J\subseteq K$. In particular, the interval is
an $|I|$-dimensional Boolean cube.
\item[(3)] Every four-cycle in this interval is a square in $X(2)$.
\end{enumerate}
\end{claim}

\begin{proof}
\noindent\emph{(1)}
Fix $J\subseteq I$ and an ordering $\pi$ of $I$ that lists $J$ first i.e. $\pi|_{[|J|]}$ is an ordering of $J$.
By \cref{cor:ordered-endpoint-surjective}, there is a path $P$
of order $\pi$ from $v$ to $x_I$.
Let $x_J$ be its vertex after $|J|$ steps.
Any $y\in[v,x_I]\cap\Lk_J(v)$ lies on a monotone path from $v$
to $x_I$. The first $\abs{J}$ directions in this path are $J$. Applying the square-switch argument of
\cref{lem:ordered-path-switching} separately before and after $y$
gives a path of order $\pi$ through $y$.
Since $x_I\in G$, this path must be $P$, so $y=x_J$.

\smallskip
\noindent\emph{(2)}
Suppose $J\subseteq K\subseteq I$. Choose an ordering of $I$
that lists first $J$, then $K\setminus J$, and finally $I\setminus K$.
By \cref{cor:ordered-endpoint-surjective}, there is a path of this order
from $v$ to $x_I$. The vertices reached after $|J|$ and $|K|$ steps
have types $J$ and $K$, respectively. Both lie in $[v,x_I]$, so
by (1) they are $x_J$ and $x_K$. The part of the path from $x_J$
to $x_K$ gives $x_J\le x_K$.
Conversely, if $x_J\le x_K$, an upward path from $x_J$ to $x_K$
can only add directions to the type, so $J\subseteq K$.
By (1), the vertices of the interval are precisely
$\set{x_J:J\subseteq I}$. The equivalence just proved identifies
their order with inclusion of subsets of $I$, giving an
$|I|$-dimensional Boolean cube.

\smallskip
\noindent\emph{(3)}
By (2), every four-cycle in $[v,x_I]$ is a two-dimensional face
of the Boolean cube. Thus, for some $J\subseteq I$ and distinct
$j_1,j_2\in I\setminus J$, its vertices can be labeled
\[
\begin{aligned}
a&=x_{J\cup\set{j_1,j_2}},
& b&=x_{J\cup\set{j_1}},\\
b'&=x_{J\cup\set{j_2}},
& c&=x_J.
\end{aligned}
\]
Complete the path $c\lessdot b\lessdot a$ to a square in $X(2)$,
and let $b''$ be its fourth vertex.
The path through $b$ uses directions $j_1,j_2$, so the other path
$c\lessdot b''\lessdot a$ uses directions $j_2,j_1$.
Hence $b''$ has type $J\cup\set{j_2}$, the same as $b'$.
Also, $v\le c\lessdot b''\lessdot a\le x_I$, so
$b''\in[v,x_I]$. By the uniqueness in (1), $b''=b'$, and hence
$\set{a,b,b',c}\in X(2)$.
\end{proof}

For $x_I\in G$, define
\[
    \varphi(x_I)
    \defeq
    \set{x_{\set{j}}\in\Lk_{\set{j}}(v):
          x_{\set{j}}\le x_I,\ j\in I}.
\]
By \cref{clm:good-interval}, this selects exactly one vertex from each part,
indexed by the elements of $I$, so $\varphi(x_I)$ is a face of $\M$.

\begin{figure}[H]
\centering
\begin{tikzpicture}[scale=.9,every node/.style={font=\small},>=Latex]
    \coordinate (e) at (0,0);
    \coordinate (a) at (1.6,0);
    \coordinate (b) at (.6,.75);
    \coordinate (c) at (0,1.6);
    \coordinate (ab) at (2.2,.75);
    \coordinate (ac) at (1.6,1.6);
    \coordinate (bc) at (.6,2.35);
    \coordinate (abc) at (2.2,2.35);
    \foreach \u/\w in {
        e/a,e/b,e/c,a/ab,a/ac,b/ab,b/bc,c/ac,c/bc,
        ab/abc,ac/abc,bc/abc}
        \draw[gray!75] (\u)--(\w);
    \foreach \p in {e,a,b,c,ab,ac,bc,abc}
        \fill[blue!65!black] (\p) circle (2pt);
    \node[below left] at (e) {$v$};
    \node[below] at (a) {$x_1$};
    \node[left] at (b) {$x_2$};
    \node[left] at (c) {$x_3$};
    \node[right] at (ab) {$x_{12}$};
    \node[right] at (ac) {$x_{13}$};
    \node[left] at (bc) {$x_{23}$};
    \node[above] at (abc) {$x_{123}$};
    \draw[->,thick] (3,1.2)--node[above] {$\varphi$}(5,1.2);
    \begin{scope}[xshift=6cm]
        \coordinate (p1) at (0,0);
        \coordinate (p2) at (2.5,0);
        \coordinate (p3) at (1.25,2.15);
        \fill[blue!8] (p1)--(p2)--(p3)--cycle;
        \draw[blue!65!black,thick] (p1)--(p2)--(p3)--cycle;
        \foreach \p in {p1,p2,p3}
            \fill[blue!65!black] (\p) circle (2pt);
        \node[below left] at (p1) {$x_1$};
        \node[below right] at (p2) {$x_2$};
        \node[above] at (p3) {$x_3$};
        \node[below] at (1.25,0) {$x_{12}$};
        \node[left] at (.58,1.08) {$x_{13}$};
        \node[right] at (1.92,1.08) {$x_{23}$};
        \node at (1.25,.85) {$x_{123}$};
    \end{scope}
\end{tikzpicture}
\caption{For a non-colliding vertex $x_{123}$, the interval below it is a
Boolean cube.  The map $\varphi$ identifies it with the face poset of the
simplex on its three level-one vertices, with $v$ corresponding to the empty
face.}
\label{fig:good-cube-embedding}
\end{figure}
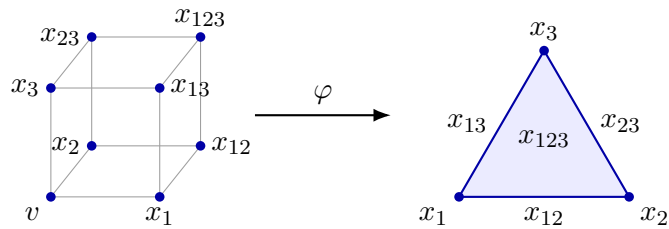

\begin{lemma}
\label{lem:poset-homomorphism}
The map $\varphi:G\to\M$ is a type-preserving poset isomorphism
between $G$ and its image.
\end{lemma}

\begin{proof}
We need to show that $\varphi$ preserves types, is injective on $G$,
and preserves the order in both directions.

\smallskip
\noindent\emph{(1) Type preservation.}
If $x_I\in G$ has type $I$, then $\varphi(x_I)$ contains exactly
one vertex from each part indexed by an element of $I$ and no other vertices.
Thus $\varphi(x_I)$ also has type $I$.

\smallskip
\noindent\emph{(2) Injectivity on $G$.}
Suppose that two distinct vertices have the same image.
By (1), they have the same type $I$; write them as $x_I,x_I'\in G$.
For every $J\subseteq I$, let $x_J$ and $x_J'$ be the unique
type-$J$ vertices in $[v,x_I]$ and $[v,x_I']$, respectively,
given by \cref{clm:good-interval}.
Choose $J$ of minimum cardinality for which $x_J\ne x_J'$.
The two intervals agree at $v$ and, since the images are equal,
at every singleton type, so $|J|\ge2$.
Choose distinct $j,j'\in J$ and put
\[
\begin{aligned}
b&\defeq x_{J\setminus\set{j}}
           =x_{J\setminus\set{j}}',\\
b'&\defeq x_{J\setminus\set{j'}}
           =x_{J\setminus\set{j'}}',\\
c&\defeq x_{J\setminus\set{j,j'}}
           =x_{J\setminus\set{j,j'}}',
\end{aligned}
\]
where the equalities follow from the choice of $J$.
By \cref{clm:good-interval}, both $\set{x_J;b,b';c}$ and
$\set{x_J';b,b';c}$ are squares in $X(2)$.
They complete the same joint formed by the edges $cb$ and $cb'$,
so uniqueness of square completion gives $x_J=x_J'$, a contradiction.
Thus $\varphi$ is injective on $G$.

\smallskip
\noindent\emph{(3) Order preservation in both directions.}
If $x_J\le x_I$, every level-one vertex below $x_J$ also lies
below $x_I$, so $\varphi(x_J)\subseteq\varphi(x_I)$.
Conversely, suppose $\varphi(x_J)\subseteq\varphi(x_I)$.
By (1), $J\subseteq I$.
Let $x_J'\in[v,x_I]$ be the unique type-$J$ vertex given by
\cref{clm:good-interval}. Since $G$ is downward closed, $x_J'\in G$.
By the forward implication and (1), both $\varphi(x_J)$ and
$\varphi(x_J')$ are contained in $\varphi(x_I)$ and have type $J$.
Since $\varphi(x_I)$ contains at most one vertex in each part,
these two faces are equal. By (2), $x_J=x_J'\le x_I$.
Therefore $\varphi$ is a type-preserving poset isomorphism onto
its image, as depicted in \cref{fig:good-cube-embedding}.
\end{proof}

\subsection{Transfer of multipartite expansion}

\begin{theorem}
\label{thm:local-expansion}
Let $X$ be a $k$-dimensional square complex and
let $v\in V_\emptyset$. Fix $0\le r\le k-2$. If
$\norm{\Lk_{\brk{k}}(v)}\ge 1-\eta$, then every cochain
$f_r:\Lk_r(v)\to\Ftwo$ satisfies
\begin{equation}
    \label{eq:local-expansion-explicit}
    \norm{\delta f_r}
    \ge
    m_{k,r}
    \dist(f_r)
    -O_k(\eta).
\end{equation}
\end{theorem}

\begin{proof}
We start by propagating the norm assumption on the top part of the link to
all its parts.  Fix $I\subseteq[k]$, $j\in I$, and put
$J\defeq I\setminus\set{j}$.  Choose an ordering of $I$ whose first
$|J|$ directions are the elements of $J$.  By
\cref{cor:ordered-endpoint-surjective},
\[
    \Lk_I(v)=\Gamma_j\prn{\Lk_J(v)}.
\]
Since every vertex has $D_j$ upward neighbours in direction $j$,
\[
    \abs{\Lk_I(v)}
    \le D_j\abs{\Lk_J(v)}.
\]
Dividing by $D_I=D_J\,D_j$, we get that
$\norm{\Lk_I(v)}\le\norm{\Lk_J(v)}$.  Iterating from $I=[k]$ down to an
arbitrary type gives
\begin{equation}
    \label{eq:top-fullness-propagates}
    \norm{\Lk_I(v)}
    \ge
    \norm{\Lk_{[k]}(v)}
    \ge 1-\eta.
\end{equation}

For $I\subseteq[k]$, let $B_I$ and $G_I$ be the bad and good vertices of
type $I$, respectively, and set
\[
    \widehat G_I\defeq\varphi(G_I)\subseteq\M(I),
    \qquad
    \widehat B_I\defeq\M(I)\setminus\widehat G_I.
\]
The exact endpoint count in \eqref{eq:endpoint-excess} gives
\[
    \norm{B_I}\le\norm{\mu_I}
    =
    1-\norm{\Lk_I(v)}
    \le \eta.
\]
The map $\varphi$ is type preserving, so it preserves the norm on $G_I$.
Thus $\norm{G_I}=\norm{\varphi(G_I)}$, and since
$\norm{\mathbf 1_{\M(I)}}=1$, we also have 
\begin{equation*}
    \norm{\widehat B_I}
    =
    1-\norm{\widehat{G}_I} 
    =
    \prn{1-\norm{\mathbf 1_{\Lk_I(v)}}}+\norm{B_I} 
    =
    1-\norm{\Lk_I(v)}+\norm{B_I}
    \le 2\eta.
\end{equation*}
For $-1\le s\le k-1$, let $B_s,G_s,\widehat B_s,\widehat G_s$ be the
corresponding unions over all types of size $s+1$. Thus
$B_s,G_s\subseteq\Lk_s(v)$ and $\widehat B_s,\widehat G_s\subseteq\M(s)$,
and we obtain
\begin{equation}
    \label{eq:bad-level-bounds}
    \begin{aligned}
    \norm{B_s}
    \le
    \binom{k}{s+1}\eta,
    \qquad
    \norm{\widehat B_s}
    \le
    2\binom{k}{s+1}\eta.
    \end{aligned}
\end{equation}

Let $\widehat h_{r-1}:\M(r-1)\to\Ftwo$. Define the pushforward
$\widehat f_r:\M(r)\to\Ftwo$ of $f_r$ and the pullback
$h_{r-1}:\Lk_{r-1}(v)\to\Ftwo$ of $\widehat h_{r-1}$ by
\[
    \widehat f_r(y)
    \defeq
    \begin{cases}
        f_r\prn{\varphi^{-1}(y)}, & y\in\widehat G_r, \\
        0, & y\in\widehat B_r.
    \end{cases}, 
    \qquad
    h_{r-1}(x)
    \defeq
    \begin{cases}
        \widehat h_{r-1}\prn{\varphi(x)}, & x\in G_{r-1}, \\
        0, & x\in B_{r-1}.
    \end{cases}
\]
We now compare $\dist_{\M}\widehat f_r$ with $\dist_{\Lk(v)}(f_r)$. If $x\in G_r$, then all children of $x$
in $\Lk(v)$ are good, and
\cref{clm:good-interval,lem:poset-homomorphism} show that $\varphi$ bijects
them with the children of $\varphi(x)$.  Thus, for $x\in G$ and for $\widehat{\delta}$ the coboundary map on $\M$,
\[
    \widehat{\delta}\widehat h_{r-1}\prn{\varphi(x)}
    =
    \delta h_{r-1}(x).
\]
Consequently,
\begin{align*}
    \norm{\widehat f_r+\widehat{\delta}\widehat h_{r-1}}
    &\ge
    \norm{\left.
        \prn{\widehat f_r+\widehat{\delta}\widehat h_{r-1}}
        \right|_{\widehat G_r}} \\
    &=
    \norm{\left.
        \prn{f_r+\delta h_{r-1}}
        \right|_{G_r}} \\
    &\ge
    \norm{f_r+\delta h_{r-1}}-\norm{B_r} \\
    &\ge
    \dist(f_r)-\norm{B_r}.
\end{align*}
Taking the minimum over $\widehat h_{r-1}$ gives
\begin{equation}
    \label{eq:formal-approximate-minimality}
    \dist_{\M}(\widehat f_r)
    \ge
    \dist(f_r)-\norm{B_r}.
\end{equation}
Therefore \cref{thm:lmm-multipartite} implies
\begin{equation}
    \label{eq:formal-coboundary-lower-bound}
    \norm{\widehat{\delta}\widehat f_r}
    \ge
    m_{k,r}\prn{\dist(f_r)-\norm{B_r}}.
\end{equation}

Likewise, if $x\in G_{r+1}$, then its children in $\Lk(v)$ are good and
correspond under $\varphi$ to the children of $\varphi(x)$, so
\[
    \widehat{\delta}\widehat f_r\prn{\varphi(x)}
    =
    \delta f_r(x).
\]
Hence,
\[
\norm{\delta f_r}
    \ge
    \norm{\left.\prn{\delta f_r}\right|_{G_{r+1}}} 
    =
    \norm{\left.\prn{\widehat{\delta}\widehat f_r}
        \right|_{\widehat G_{r+1}}} 
    \geq
    \norm{\widehat{\delta}\widehat f_r}
        -\norm{\widehat B_{r+1}} 
\] 
Combining this with \eqref{eq:formal-coboundary-lower-bound} and then using
\eqref{eq:bad-level-bounds}, we get
\begin{align*}
    \norm{\delta f_r}
    &\ge
    m_{k,r}\dist(f_r)
        -m_{k,r}\norm{B_r}
        -\norm{\widehat B_{r+1}} \\
    &\ge
    m_{k,r}\dist(f_r)
        -\prn{
            m_{k,r}\binom{k}{r+1}
            +2\binom{k}{r+2}
        }\eta \\
    &=
    m_{k,r}\dist(f_r)-O_k(\eta),
\end{align*}
as claimed.

\end{proof}

\begin{proof}[Proof of \cref{thm:link-expansion-local}]
Let $I\subseteq \brk{k}$ and let $v\in V_I$. Put $d=k-|I|$.
Consider the sub-complex $X^I$ whose relative type-$J$ part is
\[
    V_J^I\defeq V_{I\cup J},
    \qquad J\subseteq {\brk{k}\setminus I}.
\]
This is a $d$-dimensional square complex, and $v$ lies in its relative
empty part. The link of $v$ and its levels are unchanged, and the norm
in the relative slice, restricted to this link, agrees with $\norm{\cdot}_v$.

Suppose first that $0\le r\le d-2$. Applying
\cref{thm:local-expansion} to $X^I$, with dimension $d$ and link level $r$,
together with \eqref{eq:link-part-lower-bound} gives
\begin{equation}
    \label{eq:relative-link-expansion}
    \norm{\delta f}_v
    \ge
    m_{d,r}
    \dist_{\Lk(v)}(f)
    -O_k\prn{\sum_{j\in\brk{k}\setminus I}\eps_j}.
\end{equation}
Since the sum over $\brk{k}\setminus I$ is at most
$\sum_{j=1}^k\eps_j$ this yields \eqref{eq:link-expansion}.

For link level $-1$, the cochain $f$ is just a value at $v$, giving
\[
    \norm{\delta f}_v=d\norm{f}_v \geq m_{d,-1}\norm{f}_v.
\]
\end{proof}

%% file: GTL.tex
\section{Small-set coboundary expansion from square complexes}

\label{sec:small-set-expansion}

We establish small-set coboundary expansion at every level below the top, arising from imbalanced lossless square complexes. Throughout this section $X$ is a $k$-dimensional square complex with parameters $(D_i, D_i', c_i, \eps_i)_{i=1}^k$. The proof follows the heavy-branch strategy of \cite{FK24}.

Throughout this section, we assume that every set $A\subseteq V_J$ to which lossless expansion is applied in direction $j\notin J$ satisfies $\norm{A}\leq c_j\norm{V_J}$. The choice of $c$ in \eqref{eq:c} ensures this whenever the initial cochain $H_r$ satisfies $\norm{H_r}\leq c\norm{V_r}$.

\input{GTL.Decomposition}
\input{GTL.Excess}
\input{GTL.BadEdges}
\input{GTL.LightSquares}
\input{GTL.MassProp}
\input{GTL.parameters}

%% file: GTL.Decomposition.tex
\subsection{Heavy and terminal vertices}
\label{sec:ltg:decomposition}
Fix a level \(0\le r\le k-1\) and a cochain,
\[
    H_r=\bigsqcup_{\abs{I}=r}H_I\subseteq V_r.
\]
We regard \(H_r\) as the initial set of \emph{heavy} vertices. 
Following \cite{FK24}, we recursively pull \(H_r\) down through the complex. 
Suppose that the heavy vertices at level \(i\) have already been defined.    
We declare \(x\) at level $i-1$ heavy, if it has many heavy parents. Concretely,
for $x\in V_I$ and $j\notin I$, we denote the \emph{direction $j$ heavy degree} and \emph{normalized heavy degree} of \(x\) by
\begin{equation}\label{eq:heavy-degree}
    d_{j, H_{I\cup{\set{j}}}}(x)
    \defeq
    \abs{\set{y\in H_{I\cup\set{j}}:x\lessdot y}}, \qquad
    u_{j,H_{I\cup\set{j}}}(x)
    \defeq
    \frac{d_{j,H_{I\cup\set{j}}}(x)}{D_j}.
\end{equation}
When $I$ and $H_{I\cup\set{j}}$ are clear from context, we will omit $H_{I\cup\set{j}}$ from the notation.
The construction uses thresholds \(0<\gamma_i<k-i\) for \(0\le i<r\), which will be fixed later. For $\abs{I}<r$, a vertex \(x\in V_I\) is declared heavy if its total heavy normalized degree is at least the threshold of its level, $\gamma_{\abs{I}}$.
\begin{equation}
H_I
\defeq
\set{
    x\in V_I:
    \sum_{j\notin I}u_j(x)\ge\gamma_{\abs{I}}}, \qquad H_i\defeq\bigsqcup_{\abs{I}=i}H_I.
\end{equation}
We also define heavy vertices in the level \emph{above} the initial cochain,
\begin{equation}
    H_{r+1}
    \defeq
    \set{v\in V_{r+1}:a\lessdot v\text{ for some }a\in H_r},
\end{equation}
Notice that \(H_{r+1}\) is not obtained from another threshold test - it is
simply the full upward neighborhood of \(H_r\). Thus every
\(a\in H_r\) has all its upward neighbors in \(H_{r+1}\), and hence
for every $j\notin\type(a)$,
$d_{j}(a)=D_j$ and $\sum_{j \notin\type(a)} u_j(a)=k-r$.
Let
\[
H\defeq\bigcup_{i=0}^{r+1}H_i
\]
be the set of heavy vertices. All other vertices are called \emph{light}, and are denoted by $L \defeq X(0)\setminus H$. We view the heavy vertices as an oriented graph whose edges point downward to the level below. For \(v\in H_i\), define its \emph{heavy children} by
\[
\Ch(v)
\defeq
\begin{cases}
    \set{u\in H_{i-1}:u\lessdot v}, & i\ge1,\\
        \emptyset, & i=0.
    \end{cases}
\]

\begin{definition}[Terminal vertices]
A heavy vertex \(v\in H\) is \emph{terminal} if
\(\Ch(v)=\emptyset\).
We write 
\[
\Ter \defeq
\set{v\in H:\Ch(v)=\emptyset}.
\]
\end{definition}

Observe that by definition, every vertex in \(H_{r+1}\) has a heavy child, and so it cannot be terminal and  
\begin{align*}
\Ter\subseteq\bigcup_{i=0}^r H_i.    
\end{align*}
For every heavy vertex \(v\in H\), recursively define its set of
terminal descendants by
\begin{equation}
    \Ter(v)
    \defeq
    \begin{cases}
    \set{v}, & v\in \Ter,\\   \displaystyle\bigcup_{u\in\Ch(v)}\Ter(u), & 
     v\notin \Ter.
    \end{cases}
\end{equation}
In words, we start at \(v\), follow every possible downward heavy path, and
record the terminal vertices at which those paths stop. From the other perspective, for a terminal vertex $t\in \Ter$ we let its \emph{terminal cluster}
$H(t)$ be 
\begin{equation}
H(t)
\defeq
\set{a\in H:t\in\Ter(a)}.
\label{terminal:cluster}
\end{equation}
Equivalently, \(a\in H(t)\) precisely when \(t\) is a descendant of \(a\)
along a downward path of heavy vertices.
The clusters need not be disjoint: a vertex \(a\) belongs to a cluster for
each terminal in \(\Ter(a)\), hence to exactly $|\Ter(a)|$ terminal clusters.
The quantity $|\Ter(a)|-1$ measures \emph{excess} coverage by terminal clusters, we call it the \emph{excess} of $a$ and we denote it by $\Ex(a)$, i.e.,  
\begin{equation}   \Ex(a)\defeq\abs{\Ter(a)}-1.
\label{terminal:excess}
\end{equation}
We observe that every heavy vertex has at least one terminal descendant, so that $\Ex(v)$ is nonnegative. \(\Ex(v)=0\) exactly when \(v\) has a unique terminal descendant; every additional terminal endpoint contributes one unit of excess. 
We measure the total ambiguity at level \(i\) by
\[
    \norm{\Ex_i} = \sum_{a\in H_i} \Ex(a)\norm{a}.
\]
These quantities play a central role in what follows.
We refer the reader to \cref{fig:two-dimensional-square-complex} for an illustration of the setup described in this subsection.
It turns out that in the small-set range, when the initial set \(H_r\) is
small enough for lossless vertex expansion to apply, almost all vertices in \(H\) are assigned to a unique terminal vertex, and the norm of the excess function is small. Proving this is a main goal of this section.

\subsection{The global to local reduction}
\label{sec:ltg:reduction}
To prove small-set coboundary expansion we need to show that $\norm{\delta H_r}$ is large relative to $\norm{H_r}$. We decompose the global function $\delta H_r$ into a sum of local functions, associated with terminal vertices. 

We now compare $\norm{\delta H_r}$ to $\sum_{t\in \Ter}
\norm{\delta H_r(t)}$. Notice that if
\begin{enumerate}
\item
Every vertex in $H_r$ has only one terminal (which, in turn, implies that vertices in $H_r$ that have different terminals have disjoint heavy descendants at all levels below $r$), and,
\item 
Vertices in $H_r$ that have different terminals have different ancestors
at level $r+1$
\end{enumerate}
then   
$\norm{\delta H_r}=\sum_{t\in \Ter}\norm{\delta H_r(t)}$.
Thus,
there are two reasons why
$\norm{\delta H_r}$
might not equal
$\sum_{t\in \Ter}
\norm{\delta H_r(t)}$.
First, a vertex of \(H_r\)
may belong to several terminal clusters.  Second, even after this
overcounting is corrected in level \(r\), the cluster coboundaries may overlap
and cancel in level \(r+1\).  
We prove:

\begin{proposition}
[Local-to-global comparison]
\label{ltg}
\begin{equation}
    \norm{\delta H_r}
    \ge
    \sum_{t\in \Ter}\norm{\delta H_r(t)}
    -
    (k-r)\norm{\Ex_r}
    -
    2\norm{\Ex_{r+1}}.
\end{equation}
\end{proposition}

\begin{proof}
Notice that
\[
H_r
=
\sum_{t\in \Ter}H_r(t)
    +
    \sum_{a\in H_r}\Ex(a)\set{a}.
\]
where the sum is over \(\Ftwo\),  
because every
\(a\in H_r\) 
is counted $\abs{\Ter(a)}$ times in the first argument, and $\Ex(a)=\abs{\Ter(a)}-1$ in the second argument.

We first deal with the problem that a vertex $a \in H_r$ may have several terminals.
\begin{align}
\norm{\delta H_r}
\ge
\norm{\delta \sum_{t\in \Ter} H_r(t)}
-
\sum_{a\in H_r}\Ex(a)\norm{\delta\set{a}}
=
\norm{\sum_{t\in \Ter}\delta H_r(t)}
    -
    (k-r)\norm{\Ex_r},
    \label{ltg:first:error}
\end{align}
where the first inequality is the triangle inequality,
and the equality follows from
$\norm{\delta\set{a}}=(k-r)\norm{a}$ for $a\in H_r$.

We now deal with the problem that neighborhoods of vertices in $H_r$ may overlap. 
For \(b\in H_{r+1}\) and a terminal $t \in \Ter$, we say $t$ \emph{contributes} to $b$, if $b\in\delta H_r(t)$,
i.e., if $b$ belongs to the coboundary of the local piece $H_r(t)$. 
Let
\[
    n_b
    \defeq
    \abs{\set{t\in \Ter :b\in\delta H_r(t)}}.
\]
Thus \(n_b\) is the number of terminal-cluster coboundaries that contain
\(b\). 
Notice that
\[
    n_b\le\abs{\Ter(b)},
\]
and equality may not hold, because it is possible that $b$ is connected through heavy vertices to a terminal $t$, but $b$ cancels out in $\delta H_r(t)$.

In this terminology,
\[
\sum_{t\in \Ter}\norm{\delta H_r(t)}
=
\sum_{b \in H_{r+1}} n_b \norm{b}, \quad \text{and} \quad
\norm{\sum_{t\in \Ter}\delta H_r(t)}
=
\sum_{b \in H_{r+1}} (n_b \bmod 2) \norm{b}.
\]
Hence,
\begin{align*}
\sum_{t\in \Ter}\norm{\delta H_r(t)}
-
\norm{\sum_{t\in \Ter}\delta H_r(t)}
&=
\sum_{b \in H_{r+1}}(n_b- (n_b \bmod 2)) \norm{b} 
\leq \sum_{b \in H_{r+1}} 2(n_b-1)_+\norm{b}.
\end{align*}
where we used the inequality $n-(n \mod 2) \leq 2(n-1)_+$ which holds for every $n\in \mathbb{N}$. Now, substituting $n_b \le \abs{\Ter(b)} =\Ex(b)+1$ and rearranging, we get
\begin{equation}
    \sum_{t\in \Ter}\norm{\delta H_r(t)}
    \le
    \norm{\sum_{t\in \Ter}\delta H_r(t)}
    +
    2\sum_{b\in H_{r+1}} \Ex(b)\norm{b}
    =
    \norm{\sum_{t\in \Ter}\delta H_r(t)}
    +
    2\norm{\Ex_{r+1}}.
    \label{ltg:second:error}
\end{equation}
Combining \cref{ltg:first:error,ltg:second:error} proves the proposition.
\end{proof}

%% file: GTL.Excess.tex
\subsection{Excess}
\label{sec:excess}

We want to control the overall excess $\Ex_i$ through \emph{local} structures, where by local we mean that
we only consider vertices that are one layer or two away. We define:

\begin{itemize}
\item 
For \(a\in H_i\), \(i\ge1\), define the same-type excess
\begin{equation}
    \Same(a)
    \defeq
    \sum_{\abs{J}=i-1}
    \prn{\abs{\Ch(a)\cap H_J}-1}_+.
    \label{eq:same-type-excess}
\end{equation}

\item
If \(b\) and \(b'\) are children of \(a\) of distinct types, the cubical
completion property supplies a lower vertex completing their square. We say that the square is \emph{light} if its lower vertex is light, $b\meet[a]b'\in L$. We denote the number of light squares below $a$ by $\Light(a)$, i.e.,
\begin{equation}
    \Light(a)
    \defeq
    \abs{
        \set{
            \set{a;b,b';c} \in X(2) \big\vert
            b,b'\in\Ch(a), \quad c\in L
        }
    }.
    \label{eq:light-squares}
\end{equation}
\end{itemize}

We also set \(\Same(a)=\Light(a)=0\) for \(a\in H_0\), and
\(\Light(a)=0\) for \(a\in H_1\). These two quantities describe the two ways in which heavy branches can fail
to merge. The crux of the proof is that there are no other sources of terminal
excess.

\begin{lemma}[Recursive excess bound]
\label{lem:recursive-excess}
For every heavy vertex \(a\in H\),
\begin{equation}
    \Ex(a)
\le\Same(a)+\Light(a)
+
\sum_{x\in\Ch(a)}\Ex(x).
\label{eq:recursive-excess}
\end{equation}
\end{lemma}

\begin{proof}
If \(\Ch(a)=\emptyset\), then \(a\) is terminal and \(\Ex(a)=0\).
Otherwise, fix one heavy child \(b^\star\in\Ch(a)\), and partition the remaining children into three sets:

\begin{itemize}
\item Vertices in $\Ch(a)$ that have the same type as $b^\star$,
\item Vertices in $\Ch(a)$ that have a different type than $b^\star$, such that $b\meet[a]b^\star \in H$,
\item Vertices in $\Ch(a)$ that have a different type than $b^\star$, such that $b\meet[a]b^\star \in L$.
\end{itemize}
Accordingly, we define,
\begin{align*}
\Sametype_{a,b^\star}
&\defeq
\set{b\in\Ch(a)~:~b \neq b^\star, \quad \type(b)=\type(b^\star)},\\
\Lightsquare_{a,b^\star} 
&\defeq
\set{b\in \Ch(a):
\type(b) \neq \type(b^\star),
b \meet[a] b^\star \in L}.
\end{align*}
We claim that 
for every $b^\star \ne b \in \Ch(a)$,
\begin{equation}\label{eq:terminal-subtraction}
|\Ter(b)\setminus\Ter(b^\star)| 
 \leq (|\Ter(b)|-1) + 
\mathbf{1}_{b \in \Sametype_{a,b^\star}}
+ 
\mathbf{1}_{b \in \Lightsquare_{a,b^\star}}.
\end{equation}
To see this notice that 
$|\Ter(b)\setminus \Ter(b^\star)|
\le |\Ter(b)|$, and so if $b \in \Sametype_{a,b^\star}$,
or $b \in \Lightsquare_{a,b^\star}$ the inequality immediately holds.
Otherwise, the fourth vertex \(c=b\mathbin{\overset a\wedge}b^\star\) of the unique square completing \(a,b,b^\star\) is heavy. Observe that $\Ter(c) \subseteq \Ter(b) \cap \Ter(b^\star)$. It follows that
\[
\abs{\Ter(b)\setminus\Ter(b^\star)} \le \abs{\Ter(b)}-\abs{\Ter(c)} \leq \abs{\Ter(b)} - 1.
\] 
Again the inequality in \eqref{eq:terminal-subtraction} immediately holds.
As
\[
\Ter(a) 
=\Ter(b^\star)\cup 
\bigcup_{b^\star \neq b\in \Ch(a)} \prn{\Ter(b)\setminus \Ter(b^\star)},
\]
we have
\begin{align}
\label{eq:union-bound-excess}
\abs{\Ter(a)}
&\le   
\abs{\Ter(b^\star)}
+
\sum_{b^\star \ne b\in \Ch(a)}
\Ex(b)
+
|\Sametype_{a,b^\star}|+|\Lightsquare_{a,b^\star}|.
\end{align}
Clearly, $|\Sametype_{a,b^\star}| \le \Same(a)$ and
$|\Lightsquare_{a,b^\star}| \le \Light(a)$. Plugging in these bounds and subtracting one from both sides of \eqref{eq:union-bound-excess} gives
\begin{align*}
\Ex(a) &\le
\Same(a)+\Light(a)+\sum_{b \in \Ch(a)}
\Ex(b)
\end{align*}
as required.
\end{proof}

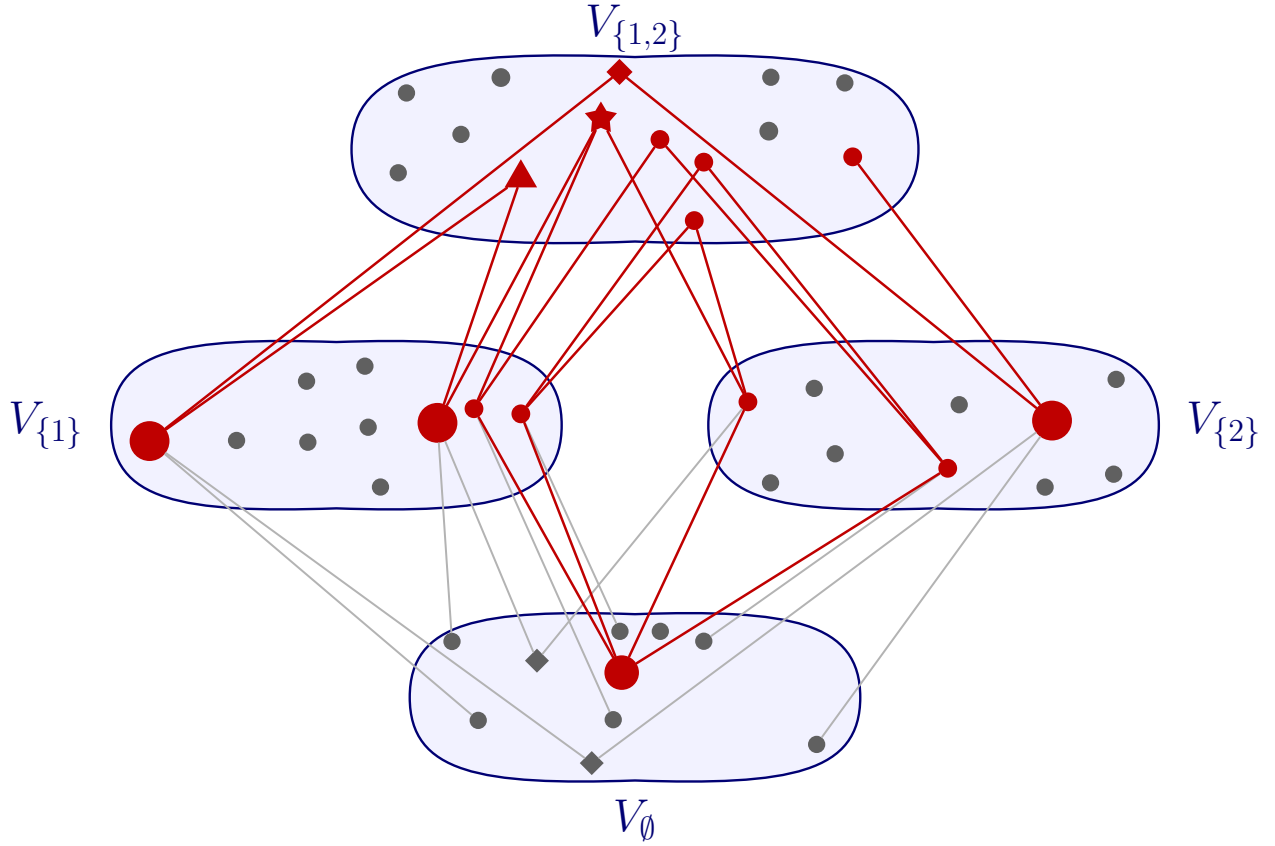
\begin{figure}[ht]
\centering
\begin{tikzpicture}[
    every node/.style={font=\small},
    blob/.style={draw=blue!45!black,fill=blue!5,line width=.9pt},
    vertex/.style={circle,fill=gray!75!black,inner sep=2.3pt},
    selected vertex/.style={circle,fill=red!75!black,inner sep=2.5pt},
    complex edge/.style={gray!60,line width=.75pt,line cap=round},
    selected complex edge/.style={red!75!black,line width=.95pt,line cap=round},
    blob label/.style={font=\Large\bfseries,text=blue!45!black}
]
    % Draw three equal lower blobs and a larger top blob, and record vertices.
    \begin{scope}[shift={(0,0)}]
        \path[blob]
            (-2.98,0)
            .. controls (-2.98,1.00) and (-2.05,1.17) .. (0,1.10)
            .. controls (2.05,1.17) and (2.98,1.00) .. (2.98,0)
            .. controls (2.98,-1.00) and (2.05,-1.17) .. (0,-1.10)
            .. controls (-2.05,-1.17) and (-2.98,-1.00) .. cycle;
        \node[vertex, shape=diamond, minimum width=6pt, minimum height=6pt] (e1) at (-.572,-.870) {};
        \node[vertex, shape=diamond, minimum width=6pt, minimum height=6pt] (e2) at (-1.297,.485) {};
        \node[vertex] (e3) at (.335,.873) {};
        \node[selected vertex, minimum width=13pt] (e4) at (-.176,.327) {};
        \node[vertex] (e5) at (-.288,-.296) {};
        \node[vertex] (e6) at (2.403,-.623) {};
        \node[vertex] (e7) at (-2.421,.739) {};
        \node[vertex] (e8) at (-2.074,-.305) {};
        \node[vertex] (e9) at (-.199,.873) {};
        \node[vertex] (e10) at (.910,.742) {};
    \end{scope}

    \begin{scope}[shift={(-3.95,3.60)}]
        \path[blob]
            (-2.98,0)
            .. controls (-2.98,1.00) and (-2.05,1.17) .. (0,1.10)
            .. controls (2.05,1.17) and (2.98,1.00) .. (2.98,0)
            .. controls (2.98,-1.00) and (2.05,-1.17) .. (0,-1.10)
            .. controls (-2.05,-1.17) and (-2.98,-1.00) .. cycle;
        \node[selected vertex, minimum width=15pt, minimum height=14pt] (a1) at (-2.472,-.210) {};
        \node[vertex] (a2) at (-.378,-.225) {};
        \node[vertex] (a3) at (.377,.781) {};
        \node[selected vertex] (a4) at (1.820,.218) {};
        \node[vertex] (a5) at (.583,-.818) {};
        \node[vertex] (a6) at (-1.321,-.202) {};
        \node[selected vertex, minimum width=15pt] (a7) at (1.337,.031) {};
        \node[vertex] (a8) at (.420,-.027) {};
        \node[selected vertex] (a9) at (2.44,0.15) {};
        \node[vertex] (a10) at (-.395,.582) {};
    \end{scope}

    \begin{scope}[shift={(3.95,3.60)}]
        \path[blob]
            (-2.98,0)
            .. controls (-2.98,1.00) and (-2.05,1.17) .. (0,1.10)
            .. controls (2.05,1.17) and (2.98,1.00) .. (2.98,0)
            .. controls (2.98,-1.00) and (2.05,-1.17) .. (0,-1.10)
            .. controls (-2.05,-1.17) and (-2.98,-1.00) .. cycle;
        \node[vertex] (b1) at (-1.302,-.378) {};
        \node[selected vertex] (b2) at (-2.456,.309) {};
        \node[vertex] (b3) at (2.382,-.648) {};
        \node[vertex] (b4) at (2.415,.606) {};
        \node[vertex] (b5) at (-2.157,-.763) {};
        \node[selected vertex, minimum width=15pt] (b6) at (1.567,.061) {};
        \node[vertex] (b7) at (-1.579,.487) {};
        \node[selected vertex] (b8) at (.187,-.570) {};
        \node[vertex] (b9) at (1.475,-.819) {};
        \node[vertex] (b10) at (.339,.271) {};
    \end{scope}

    \begin{scope}[shift={(0,7.25)}]
        \path[blob]
            (-3.75,0)
            .. controls (-3.75,1.10) and (-2.50,1.32) .. (0,1.22)
            .. controls (2.50,1.32) and (3.75,1.10) .. (3.75,0)
            .. controls (3.75,-1.10) and (2.50,-1.32) .. (0,-1.22)
            .. controls (-2.50,-1.32) and (-3.75,-1.10) .. cycle;
        \node[vertex] (t1) at (-3.023,.745) {};
        \node[vertex] (t2) at (1.797,.953) {};
        \node[selected vertex, shape=star, star point ratio=1.8] (t3) at (-.451,.409) {};
        \node[selected vertex] (t4) at (.784,-.943) {};
        \node[selected vertex, fill=gray!75!black] (t5) at (1.77,0.24) {};
        \node[selected vertex] (t6) at (2.880,-.100) {};
        \node[selected vertex, shape=regular polygon, regular polygon sides=3] (t7) at (-1.51,-0.38) {};
        \node[selected vertex, fill=gray!75!black] (t8) at (-1.774,.949) {};
        \node[selected vertex] (t10) at (0.91,-0.17) {};
        \node[selected vertex, shape=diamond, minimum width=6pt, minimum height=6pt] (t11) at (-.205,1.023) {};
        \node[selected vertex] (t12) at (0.33,0.13) {};
        \node[vertex] (t13) at (-3.132,-.311) {};
        \node[vertex] (t14) at (2.776,.878) {};
        \node[vertex] (t15) at (-2.302,.197) {};
    \end{scope}

    % Explicit gray edges.
    \draw[complex edge] (a1)--(e1);
    \draw[complex edge] (a1)--(e8);
    \draw[complex edge] (a4)--(e5);
    \draw[complex edge] (a7)--(e2);
    \draw[complex edge] (a7)--(e7);
    \draw[complex edge] (a9)--(e9);
    \draw[complex edge] (b2)--(e2);
    \draw[complex edge] (b6)--(e6);
    \draw[complex edge] (b6)--(e1);
    \draw[complex edge] (b8)--(e10);

    % Explicit red edges.
    \draw[selected complex edge] (a4)--(e4);
    \draw[selected complex edge] (a9)--(e4);
    \draw[selected complex edge] (b2)--(e4);
    \draw[selected complex edge] (b8)--(e4);
    \draw[selected complex edge] (a1)--(t7);
    \draw[selected complex edge] (a1)--(t11);
    \draw[selected complex edge] (a4)--(t3);
    \draw[selected complex edge] (a4)--(t12);
    \draw[selected complex edge] (a7)--(t3);
    \draw[selected complex edge] (a7)--(t7);
    \draw[selected complex edge] (a9)--(t4);
    \draw[selected complex edge] (a9)--(t10);
    \draw[selected complex edge] (b2)--(t4);
    \draw[selected complex edge] (b2)--(t3);
    \draw[selected complex edge] (b6)--(t6);
    \draw[selected complex edge] (b6)--(t11);
    \draw[selected complex edge] (b8)--(t12);
    \draw[selected complex edge] (b8)--(t10);

    % Labels face outward, away from all incidence edges.
    \node[blob label] at (0,-1.62) {$V_{\emptyset}$};
    \node[blob label,anchor=east] at (-7.15,3.60) {$V_{\{1\}}$};
    \node[blob label,anchor=west] at (7.15,3.60) {$V_{\{2\}}$};
    \node[blob label] at (0,8.85) {$V_{\{1,2\}}$};
\end{tikzpicture}
\caption{The four vertex sets of a two-dimensional square complex, arranged by level. An initial set of vertices $H_1$ is marked heavy (red) on level $1$. In the layer below, only vertices that have many heavy parents are heavy. In the layer above, all upward neighbors of the heavy vertices are marked heavy. Terminal vertices (heavy vertices without heavy children) are enlarged. Red vertices shaped like diamonds, stars, and triangles have excess. The excess of the triangular vertex is accounted for by $\Same_2$, whereas the excess of the diamond shaped vertex is accounted for by $\Light_2$. The excess of the star shaped vertex is accounted for by both functions.}
\label{fig:two-dimensional-square-complex}
\end{figure}

In words, the recursive lemma says that ambiguity can only come from
ambiguity already present one layer below, together with new same-type and
light-square charges created at level \(i\). We now unroll the recursion.

\begin{theorem}[Excess recursion]
\label{thm:flag-norm-excess}
For every \(1\le i\le r+1\),
\begin{equation}
\norm{\Ex_i}
\le
(k-i+1)\norm{\Ex_{i-1}}
+
\norm{\Same_i}
+
\norm{\Light_i}
\label{eq:flag-norm-excess}
\end{equation}
\end{theorem}

\begin{proof}
By \cref{lem:recursive-excess},
\begin{align*}
\sum_{a \in H_i}\norm{a}\Ex(a)
&\le
\sum_{a \in H_i}\norm{a}\prn{\Same(a)+\Light(a)}+\sum_{a \in H_i}\norm{a}\sum_{b \in \Ch(a)}\Ex(b).
\end{align*}
Rearranging the first sum, we have
\[
    \sum_{a \in H_i}\norm{a}\prn{\Same(a)+\Light(a)} = \norm{\Same_i} + \norm{\Light_i}.
\]For the second sum, changing the order of summation, we get
\begin{align*}
\sum_{a\in H_i}\norm{a}\sum_{b\in\Ch(a)}\Ex(b)
=
\sum_{b\in H_{i-1}}
\Ex(b)
\sum_{a:b\in\Ch(a)}\norm{a}
\le
\sum_{b\in H_{i-1}}\Ex(b)
\sum_{b\lessdot a}\norm{a},
\end{align*}
where the inequality holds because on the left-hand side we sum only over heavy parents, whereas on the right-hand side we sum over all parents. Next, we use the norm-identity on \cref{eq:norm-propagation} to obtain
\[
\sum_{b\in H_{i-1}}\Ex(b)
\sum_{b\lessdot a}\norm{a}
=
(k-i+1)\sum_{b\in H_{i-1}}\Ex(b)\norm{b}
= (k-i+1)\norm{\Ex_{i-1}}.
\]
Combining the equations above concludes the proof.
\end{proof}
Upon recalling that every $a\in H_0$ is terminal and has no excess, we get $\Ex_0 = 0$ and we can immediately deduce the following:
\begin{corollary}[Iterated excess bound]
\label{cor:iterated-excess}
For every \(0\le i\le r+1\),
\begin{equation}
\norm{\Ex_i}\le
\sum_{j=1}^{i}\frac{(k-j)!}{(k-i)!}
\prn{\norm{\Same_j} + \norm{\Light_j}}.
\label{eq:iterated-excess}
\end{equation}
\end{corollary}

%% file: GTL.BadEdges.tex
\subsection{Same Type Collisions}
\label{sec:same-type}
The previous subsection reduced terminal overlap to the level charges $\norm{\Same_i}$ and $\norm{\Light_i}$.
Our next task is to show that every one of these charges is small compared to the
norm of the set \(H_{i-1}\). There are two ingredients:
lossless expansion controls same-type collisions and additional imbalance assumptions control light squares. In this subsection we deal with the first type of error.

\begin{lemma}
\label{lem:same-norm-bound}
For every \(1\le i\le r+1\),
\begin{equation}
    \norm{\Same_i}
    \le
    \prn{\sum_{j=1}^k\eps_j}\norm{H_{i-1}}.
\end{equation}
\end{lemma}

\begin{proof}
Expand the definition of \(\norm{\Same_i}\) and group the sum by the type \(I\) of the heavy vertices at level \(i\):
\begin{align*}
\sum_{a\in H_i}\norm{a}\Same(a)
&=
\sum_{\abs{I}=i}\norm{I}\cdot \sum_{a \in H_I} \sum_{j\in I} \prn{\deg_{H_{I\setminus\set{j}}}(a)-1}_+ \\
&=
\sum_{\abs{I}=i}\norm{I}\cdot \sum_{j\in I}
\sum_{a \in H_I} 
\prn{\deg_{H_{I\setminus\set{j}}}(a)-1}_+\\
&\le
\sum_{\abs{I}=i}
\norm{I}\cdot \sum_{j\in I}
\eps_j D_j \abs{H_{I\setminus\set{j}}},
\end{align*}
where we have used \cref{lossless:excess} in \cref{lem:lossless-estimates}.

Each set $J$ of size $i-1$ appears in the summands corresponding to $I=J\cup\set{j}$ for $j\notin J$. Therefore, we can rewrite the sum as
\[
    \norm{\Same_i}
    \le
    \sum_{\abs{J}=i-1} \sum_{j\notin J} \norm{J\cup\set{j}} \eps_j D_j \abs{H_J}
    = 
    \sum_{\abs{J}=i-1} \sum_{j\notin J} \eps_j \norm{H_J}
    \leq
    \prn{\sum_{j=1}^k \eps_j} \norm{H_{i-1}}
\]
\end{proof}

%% file: GTL.LightSquares.tex
\subsection{Light Squares}
\label{sec:light:squares}

Each light square contributing to \(\Light_i\) lies between levels
\(i-2\), \(i-1\), and \(i\).  Instead of counting these squares from their
heavy upper vertices at level \(i\), we count them from their light lower
vertices at level \(i-2\). Throughout, we assume that every set to which lossless expansion is applied is small enough to be in the appropriate lossless range.

\begin{lemma}
\label{lem:light-squares}
For every \(2\le i\le r+1\),
\begin{equation}
\norm{\Light_i}
\le
(i-1)(k-i+1)\max_{s<t}\set{\eps_{t}D_{s}',\gamma_{i-2}}\norm{H_{i-1}}.
\label{eq:light-square-bound}
\end{equation}
\end{lemma}

\begin{proof}
Fix \(2\le i\le r+1\). By definition, \[
\norm{\Light_i}
=
\sum_{a\in H_i}\abs{\set{\set{a;b,b';c}\in X(2):b,b'\in \Ch(a),\ c \in L}}\cdot\norm{a}
= 
\sum_{a\in H_i} \sum_{\substack{\set{b,b'}\subseteq \Ch(a)\\ \type(b)\neq\type(b')}} \mathbf{1}_{b\meet[a]b' \in L} \norm{a}.
\]
Instead of counting the light squares from above, we change the order of summation and count them from below.
\[
\norm{\Light_i}
=
\sum_{c\in L_{i-2}}\sum_{\substack{\set{b,b'}\subseteq H_{i-1} \\ \type(b)\neq\type(b') \\c\lessdot b,b'}}\mathbf{1}_{b\meet[c]b' \in H_{i}}\norm{b\meet[c]b'}
\leq
\sum_{c\in L_{i-2}}\sum_{\substack{\set{b,b'}\subseteq H_{i-1}\\ \type(b)\neq\type(b') \\ c\lessdot b,b'}}\norm{b\meet[c]b'}.
\]
For a fixed \(c\in L_J\), unique square completion gives one square for
every choice of heavy parents in two directions $s<t$ outside $J$. Consequently,
\[
\begin{aligned}
\sum_{\substack{\set{b,b'}\subseteq H_{i-1} \\ \type(b)\neq\type(b') \\ c\lessdot b,b'}}\norm{b\meet[c]b'}
&=
\sum_{\substack{s<t\\s,t\notin J}}
\norm{J\cup\set{s,t}}
d_{s}(c)
d_{t}(c)\\
&=
\norm{J}
\sum_{\substack{s<t\\s,t\notin J}}u_s(c)u_t(c).
\end{aligned}
\]
Substituting this identity into the preceding inequality gives
\begin{equation}
\norm{\Light_i}
\le
\sum_{\abs{J}=i-2}
\norm{J}
\sum_{c\in L_J}
\sum_{\substack{s<t\\s,t\notin J}}u_s(c)u_t(c).
\label{eq:light-counted-from-below}
\end{equation}

For a fixed $J$, we analyze the inner two sums of \eqref{eq:light-counted-from-below}. List the directions outside \(J\) as
\(m_1<m_2<\cdots<m_{k-i+2}\) and rewrite the pair-product sum as
\[
\sum_{1\le s<t\le k-i+2}u_{m_s}(c)u_{m_t}(c)
=
\sum_{t=2}^{k-i+2}
u_{m_t}(c)U_{<m_t}(c), \qquad \text{for} \qquad U_{<m_t}(c)\defeq\sum_{s<t}u_{m_s}(c).
\]
Following \cite{PK22}, we supply a bound for the product $u_{m_t}(c)U_{<m_t}(c)$, which is linear in the upward heavy-degrees. We claim
\begin{equation}\label{eq:linearization}
    u_{m_t}(c)U_{<m_t}(c) \leq \gamma_{i-2}\prn{u_{m_t}(c)-\eps_{m_t}}_+ + \eps_{m_t}U_{<m_t}(c).
\end{equation}
Indeed, for $u_{m_t}(c) \leq \eps_{m_t}$, the second term on the right-hand side already dominates the left-hand side. For $u_{m_t}(c) > \eps_{m_t}$, the inequality follows by plugging in $\gamma_{i-2} \geq U_{<m_t}(c)$, which follows from the fact that $c$ is light.
Applying this bound pointwise for every $c$ in the summation and changing the order of summation, we obtain
\[
\sum_{c\in L_J}
\sum_{\substack{s<t\\s,t\notin J}}u_s(c)u_t(c)
\leq
\sum_{t=2}^{k-i+2}
\brk{\gamma_{i-2}\sum_{c\in L_J}\prn{u_{m_t}(c)-\eps_{m_t}}_+ + \eps_{m_t}\sum_{c\in L_J}U_{<m_t}(c)}.
\]
Now, once the expression is linear in the upward-degrees, we translate the summation to edge counts. \cref{lossless:large:degree} in \cref{lem:lossless-estimates} and biregularity give, for every 
\(m\notin J\),
\[
\sum_{c\in L_J}
\prn{u_m(c)-\eps_m}_+
\le
\frac{\abs{H_{J\cup\set{m}}}}{D_m}
, \qquad 
\sum_{c\in L_J}u_m(c)
\le
\frac{D_m'\abs{H_{J\cup\set{m}}}}{D_m}.
\]
Combining these two estimates with \eqref{eq:linearization}, we get
\begin{align*}
\sum_{c\in L_J}
\sum_{\substack{s<t\\s,t\notin J}}u_s(c)u_t(c) 
&\le
\sum_{t=2}^{k-i+2}
\left(
\gamma_{i-2}\frac{\abs{H_{J\cup\set{m_t}}}}{D_{m_t}}
+
\sum_{s < t}
\eps_{m_t}D_{m_s}'\frac{\abs{H_{J\cup\set{m_s}}}}{D_{m_s}}
\right).
\end{align*}
Rearranging the summation, the coefficient of $\abs{H_{J\cup\set{m_s}}}$ is
\[
    \frac{1}{D_{m_s}}\prn{\mathbf{1}_{s\ge2}\gamma_{i-2} + \sum_{t=s+1}^{k-i+2}\eps_{m_t}D_{m_s}'} \leq \frac{k-i+1}{D_{m_s}}\max_{1\le p<q\le k}\set{\eps_qD_p',\gamma_{i-2}}.
\]
Indeed, the expression in parentheses has at most $k-i+1$ nonzero terms, each at most $\max_{1\le p<q\le k}\set{\eps_qD_p',\gamma_{i-2}}$.
Substituting this bound into \eqref{eq:light-counted-from-below}, we obtain
\begin{align*}
\norm{\Light_i}
&\le
(k-i+1)\max_{s<t}\set{\eps_tD_s',\gamma_{i-2}}
\sum_{\abs{J}=i-2}\norm{J}
\sum_{m\notin J}
\frac{\abs{H_{J\cup\set{m}}}}{D_m} \\
&= (k-i+1)\max_{s<t}\set{\eps_tD_s',\gamma_{i-2}} \sum_{\abs{J} = i-2} \sum_{m\notin J} \norm{H_{J \cup \set{m}}}.
\end{align*}
Every type \(K\) of size \(i-1\) arises from exactly \(i-1\) pairs
\((J,m)\), one for each \(m\in K\).  Consequently,
\[
\norm{\Light_i}
\le
(i-1)(k-i+1)\max_{s<t}\set{\eps_{t}D_s',\gamma_{i-2}}\norm{H_{i-1}},
\]
as claimed.
\end{proof}

%% file: GTL.MassProp.tex
\subsection{Heavy-mass propagation}\label{sec:mass-prop}
In \cref{sec:same-type,sec:light:squares} we showed that $\norm{\Same_i}$ and $\norm{\Light_i}$ are small relative to $\norm{H_{i-1}}$. However, we want to compare the error with $\norm{H_r}$ rather than $\norm{H_{i-1}}$. In this section we relate the two quantities.
\begin{lemma}[Heavy-mass propagation]
For every \(1\le i\le r\), if $\gamma_{i-1} > \sum_{j\in\brk{k}}\eps_j$, then
\begin{equation}
    \norm{H_{i-1}}
    \le
    \frac{i}{\gamma_{i-1} - \sum_{j \in \brk{k}} \eps_j}
    \norm{H_i}.
    \label{eq:one-step-heavy-mass}
\end{equation}
The same holds for $i=r+1$ under the convention $\gamma_r = k-r$.
\end{lemma}

\begin{proof}
Fix $1\leq i\leq r$ and \(I\subseteq\brk{k}\) with \(\abs{I}=i-1\). The
definition of heaviness gives
\(\sum_{j\notin I}u_j(a)\ge\gamma_{i-1}\) for every \(a\in H_I\). Summing over \(a\in H_I\), we obtain
\begin{equation}
    \gamma_{i-1}\abs{H_I}
    \le
    \sum_{j\notin I}\frac{1}{D_j}
    \sum_{a\in H_I}d_j(a).
    \label{eq:heavy-mass-lower-bound}
\end{equation}
For each \(j\notin I\), the lossless edge estimate gives
\[
    \sum_{a\in H_I}d_j(a)
    =
    e\prn{H_I,H_{I\cup\set{j}}}
    \le
    \abs{H_{I\cup\set{j}}}
    +
    \eps_jD_j\abs{H_I}.
\]
Substituting these estimates into
\eqref{eq:heavy-mass-lower-bound} and using
\(\sum_{j\notin I}\eps_j\le\sum_{j\in\brk{k}}\eps_j\) yields
\begin{equation}
    \prn{\gamma_{i-1}-\sum_{j\in\brk{k}}\eps_j}\abs{H_I}
    \le
    \sum_{j\notin I}
    \frac{\abs{H_{I\cup\set{j}}}}{D_j}.
    \label{eq:heavy-mass-by-type}
\end{equation}
Divide both sides by $D_I$ and sum over all
types \(I\) of size \(i-1\).  Every type \(K\) of size \(i\) occurs
exactly \(i\) times on the right, once for each choice of \(j\in K\).
Consequently,
\[
    \prn{\gamma_{i-1}-\sum_{j\in\brk{k}}\eps_j}\norm{H_{i-1}}
    \le
    i\norm{H_i}.
\]
Dividing both sides by \(\gamma_{i-1}-\sum_{j\in\brk{k}}\eps_j\) gives the desired inequality.
\end{proof}

The iterated one-step comparison tells us the total price of pulling
heaviness down from level \(r\) to level \(i-1\). Thus, for every
\(1\le i\le r+1\), including the empty-product case \(i=r+1\), we obtain
\begin{corollary}\label{cor:heavy-mass-prop}
    Assume that \(\gamma_{s-1}>\sum_{j\in\brk{k}}\eps_j\) for every \(i\le s\le r\). Then,
\begin{equation}
    \norm{H_{i-1}}
    \le
    \prod_{s=i}^{r}\frac{s}{\gamma_{s-1} - \sum_{j \in \brk{k}} \eps_j}
    \norm{H_r}
    \label{eq:heavy-mass-propagated}
\end{equation}    
\end{corollary}

%% file: GTL.parameters.tex
\subsection{Choice of parameters and final error bound}
\label{sec:parameters}

\paragraph{Imbalance}
Throughout, we assume that the complex is imbalanced enough, so that
\begin{equation}\label{eq:parameter-imbalance}
    \max_{1\le s<t\le k}\set{\eps_tD_s'} \le \eps_\star, \qquad
    \text{for} \qquad \eps_\star \defeq \sum_{j=1}^k \eps_j.
\end{equation}

\paragraph{Initial norm bound}
Lossless expansion in direction $j\notin J$ applies to sets $A\subseteq V_J$ with $\abs{A}\leq c_j\abs{V_J}$, equivalently $\norm{A}\leq c_j\norm{V_J}$. To ensure that all sets used in the lossless estimates lie in this range, put
\begin{equation}\label{eq:c}
    c = \frac{\min_j c_j}{2^k k! D'_{\brk{k}}D_{\brk{k}}}.
\end{equation}
The reader may verify that this conservative choice of \(c\) guarantees that all sets used in the lossless expansion estimates satisfy \(\norm{A}\leq c_j\norm{V_J}\). One can also improve this constant by following the heavy-mass propagation through the levels of the complex, but we do not pursue this here.

\paragraph{Heaviness thresholds}
For \(1\le s\le r\), define
\begin{equation}
    \gamma_{s-1}
    \defeq
    \eps_\star+\eps_\star^{\nicefrac{1}{2^s}}.
    \label{eq:threshold-choice}
\end{equation}
The thresholds satisfy \(\gamma_{s-1}>\eps_\star\) for \(1\le s\le r\), as required for the heavy-mass
propagation corollary, \cref{cor:heavy-mass-prop}.

\paragraph{The main theorem}

\begin{theorem}
\label{prop:small-set-cosystolic-expansion}
    Let $X$ be a lossless $k$-dimensional square complex with parameters $\prn{D_i, D_i', c_i, \eps_i}_{i\in\brk{k}}$. Assume that the complex is imbalanced enough, so that \cref{eq:parameter-imbalance} holds. For $0\leq r\leq k-1$, let $H_r:V_r\to \Ftwo$ be a locally minimal cochain, of norm at most $c \norm{V_r}$, for $c$ as in \cref{eq:c}. Then,
    \[
        \norm{\delta H_r} \geq \prn{m_{k,r-1}-O_k\prn{\eps_\star^{\nicefrac{1}{2^r}}}}\norm{H_r}.
    \]
\end{theorem} 

\begin{proof}
    Set the thresholds as in \eqref{eq:threshold-choice}.
    Invoking \cref{ltg}, we get
    \begin{equation}\label{eq:ltg-prop}
        \norm{\delta H_r} \geq \sum_{t\in \Ter} \norm{\delta H_r(t)} - (k-r)\norm{\Ex_r} - 2\norm{\Ex_{r+1}}.
    \end{equation}
    Then, the excess norm is bounded using \cref{cor:iterated-excess}, giving us
    \begin{equation}\label{eq:excess-bound-prop}
        (k-r)\norm{\Ex_r} + 2\norm{\Ex_{r+1}} \leq O_k\prn{\norm{\Same} + \norm{\Light}}.
    \end{equation}
    The imbalance assumption gives us that for all $2\leq i\leq r+1$,
    \[
        \max_{1\le s<t\le k}\set{\eps_tD_s'} \le \eps_\star \leq \gamma_{i-2}.
    \]
    Furthermore, we may assume that $\eps_\star$ is sufficiently small in terms of $k$, since otherwise the claimed lower bound is nonpositive. Hence, all conditions of \cref{lem:same-norm-bound,lem:light-squares} are satisfied. We get
    \begin{align*}
        \norm{\Same_i} &= O_k\prn{\eps_\star}\norm{H_{i-1}}, \qquad 1\leq i\leq r+1, \\
        \norm{\Light_i} &= O_k\prn{\gamma_{i-2}}\norm{H_{i-1}}, \qquad 2\leq i\leq r+1.
    \end{align*}
    Also, $\Light_1=0$, since there are no squares with upper vertex at level $1$.
    Plugging the chosen thresholds into \cref{cor:heavy-mass-prop}, we translate the above bounds to depend on the initial set $H_r$, rather than on $H_{i-1}$. For $1\leq i\leq r+1$, including the empty-sum case $i=r+1$, we have
    \[
    \norm{H_{i-1}} = O_k\prn{\eps_\star^{-\sum_{s=i}^r 2^{-s}}}\norm{H_r} = O_k\prn{\eps_\star^{-\nicefrac{1}{2^{i-1}}+\nicefrac{1}{2^r}}}\norm{H_r}.
    \]
    Consequently, for $1\leq i\leq r+1$,
    \begin{align*}
        \norm{\Same_i} &= O_k\prn{\eps_\star^{1-\nicefrac{1}{2^{i-1}}+\nicefrac{1}{2^r}}}\norm{H_r} \leq O_k\prn{\eps_\star^{\nicefrac{1}{2^r}}}\norm{H_r}, && 1 \leq i \leq r+1, \\
        \norm{\Light_i} &= O_k\prn{\gamma_{i-2}\eps_\star^{-\nicefrac{1}{2^{i-1}}+\nicefrac{1}{2^r}}}\norm{H_r} = O_k\prn{\eps_\star^{\nicefrac{1}{2^r}}}\norm{H_r}, && 2 \leq i \leq r+1.
    \end{align*}
    Summing over these ranges gives a factor of at most $r+1$, which is absorbed into the $O_k$ notation, so that
    \[
        \norm{\Same} = O_k\prn{\eps_\star^{\nicefrac{1}{2^r}}}\norm{H_r} \qquad \text{and} \qquad \norm{\Light} = O_k\prn{\eps_\star^{\nicefrac{1}{2^r}}}\norm{H_r}.
    \]
    Combining these bounds with \cref{eq:excess-bound-prop,eq:ltg-prop} then gives us
    \begin{equation}\label{eq:ltg-prop-2}
        \norm{\delta H_r} \geq \sum_{t\in \Ter} \norm{\delta H_r(t)} - O_k\prn{\eps_\star^{\nicefrac{1}{2^r}}}\norm{H_r}.
    \end{equation}
    Fix a terminal vertex $t\in V_I$ and put $i=\abs{I}$. Viewed as a
    cochain on $\Lk(t)$, the cluster $H_r(t)$ lies at link level $r-i-1$.
    This cochain is minimal in the link. Indeed, a link coboundary extends to an ambient coboundary of a cochain supported in $\Lk(t)$. Since $H_r(t)\subseteq H_r$, any such correction decreasing the norm of $H_r(t)$ would also decrease the norm of $H_r$, contradicting local minimality.
    Applying \cref{thm:link-expansion-local} at this level and
    converting link norms to ambient norms by the factor $D_I^{-1}$ gives
    \[
        \norm{\delta H_r(t)}
        \geq
        m_{k-i,r-i-1}\norm{H_r(t)}-D_I^{-1}O_k\prn{\eps_\star} \geq m_{k,r-1}\norm{H_r(t)}-D_I^{-1}O_k\prn{\eps_\star}.
    \]
    We now show that the additive term is negligible compared with the mass of
    the terminal cluster. The heavy mass propagates, and \cref{cor:heavy-mass-prop} invoked with $H(t)$ instead of $H$ gives
    \begin{equation}
        \norm{H_r(t)}
        \geq
        \frac{i!}{r!}
        \eps_\star^{\nicefrac{1}{2^i}-\nicefrac{1}{2^r}}\norm{t}.
        % \label{eq:terminal-cluster-lower-bound}
    \end{equation}
    As $\norm{t}=D_I^{-1}$, the error term in $D_I^{-1}O_k\prn{\eps_\star}$ is bounded by
    \[
        D_I^{-1}O_k\prn{\eps_\star}
        \leq
        O_k\prn{\eps_\star^{\,1-\nicefrac{1}{2^i}+\nicefrac{1}{2^r}}}\norm{H_r(t)}
        \leq
        O_k\prn{\eps_\star^{\,\nicefrac{1}{2^r}}}\norm{H_r(t)}.
    \]
    Hence, we can replace the additive error by a multiplicative error and obtain
    \[
        \norm{\delta H_r(t)}\geq \prn{m_{k,r-1}-O_k\prn{\eps_\star^{\nicefrac{1}{2^r}}}}\norm{H_r(t)}.
    \]
    Summing over all terminals yields
    \[
        \sum_{t\in\Ter}\norm{\delta H_r(t)}
        \geq
        \prn{m_{k,r-1}-O_k\prn{\eps_\star^{\nicefrac{1}{2^r}}}}
        \sum_{t\in\Ter}\norm{H_r(t)}
        \geq
        \prn{m_{k,r-1}-O_k\prn{\eps_\star^{\nicefrac{1}{2^r}}}}
        \norm{H_r},
    \]
    where the last inequality follows because every $x\in H_r$ belongs to
    the cluster of at least one terminal. Plugging this into
    \cref{eq:ltg-prop-2} concludes the proof.
\end{proof}

\begin{proof}[Proof of \cref{thm:arbitrary-rank-expansion}]
With $c$ as in \eqref{eq:c}, \cref{prop:small-set-cosystolic-expansion}
implies that, at every level $0\leq r<k$, every locally minimal cochain
$H_r$ of norm at most $c\norm{V_r}$ expands by a factor of at least
$m_{k,r-1}-O_k\prn{\eps_\star^{\Omega_k(1)}}$. Choosing $a = \min_{r}m_{k,r-1}$ concludes the proof.
\end{proof}

%% file: Qltc.tex
\section{qLTC}

In this section we use \cref{prop:small-set-cosystolic-expansion} to show that a sufficiently imbalanced four-dimensional decorated square complex gives rise to a quantum locally testable code. We start with the required definitions and then describe the construction and prove its qLTC properties. Throughout this section, the degree and lossless-expansion parameters are fixed independently of the code length.

\subsection{Definitions}

One way to characterize a linear code is by a parity-check matrix for it, $\ker (H)=\calC$. Such a matrix is not unique. Any matrix having the same row-span over $
\FF_2$ characterizes the same code. Consider the following probabilistic procedure for testing whether a given word $w\in\FF_2^n$ is in the code: choose a row $R$ uniformly at random from $H$ and check if $R\cdot w^t=\sum_{i=1}^n R_iw_i$ is $0$ over $\FF_2$. If it is -- accept, announcing $w\in\calC$. If it is $1$ -- reject, announcing $w\notin\calC$. A codeword will never be rejected, but the rejection probability of a non-codeword depends on $H$. Ideally the distance of the given $w$ from $\calC$ should be reflected by the rejection probability of the testing procedure. Codes which allow for such a choice of $H$ are called locally testable codes.

For the codes arising from our complexes, we use the following weighted formulation of local testability. Words and syndromes carry the norms of their respective levels in the underlying complex; for the reversed complex, we use its reversed norms.

\begin{definition}[Classical Locally Testable Code (LTC) {\cite[Definition 11]{GS02,leverrier2022towardsqltc}}] \label{def:ltc}
    For $s>0$, a classical linear code $\calC\leq\FF_2^n$ is $s$-locally testable if it has a parity-check matrix $H\in\FF_2^{m\times n}$ ($\ker(H)=\calC$) with rows and columns of weight bounded by a constant such that for every word $w\in\FF_2^n$,
    \[\norm{Hw}\geq s\cdot\dist(w,\calC)\]
    where, recall,  $\dist(x,\calC)=\min_{c\in\calC}\norm{x-c}$.
\end{definition}

The parameter $s$ is the soundness of the test, while its locality is the maximum row weight of $H$.

The weighted norms in \cref{def:ltc} are comparable to Hamming weight by \cref{rem:hamming-vs-norm}. For our codes, biregularity also bounds the ratio $m/n$ above and below by positive constants depending only on $k,D_i,D_i'$. Consequently, the weighted formulation is equivalent, up to constant factors in soundness, to the usual inequality $\frac{1}{m}\abs{Hw}\geq\frac{s}{n}\dist_H(w,\calC)$, where $\dist_H$ denotes Hamming distance. These factors depend only on $k,D_i,D_i'$.

\begin{definition}[Quantum CSS code \cite{calderbank1996goodquantumeccexist,steane1996qecc}]
    A quantum CSS code $\calC(\calC_X,\calC_Z)$ is specified by two classical linear
codes $\calC_X,\calC_Z\leq\FF_2^n$ satisfying
$\calC_X^\perp\leq\calC_Z$ (and thus also $\calC_Z^\perp\leq\calC_X$).
\end{definition}

Using parity-check matrices $H_Z\in\FF_2^{m_Z\times n},H_X\in\FF_2^{m_X\times n}$ for $\calC_Z,\calC_X$ respectively, with code length $n$ and number of $Z$-checks $m_Z$ and $X$-checks $m_X$, the definition is equivalent to the following being a cochain complex (composition of maps is zero): $\FF_2^{m_X}\xrightarrow{H_X^t}\FF_2^n\xrightarrow{H_Z}\FF_2^{m_Z}$. We say the middle space $\FF_2^n$ hosts the qubits.

Using this notation, the distance of the CSS code is the minimum $\min\{d_Z,d_X\}$ with $d_Z\coloneq\min\{\abs{w}:w\in\ker(H_Z)\setminus\im(H_X^t)\}$ and $d_X=\min\{\abs{w}:w\in\ker(H_X)\setminus\im(H_Z^t)\}$.
The rate $\rho$ is the ratio of log-dimension of the code to the number of qubits $n$ (its length). It is lower bounded by $\rho\geq\frac{n-m_Z-m_X}{n}=1-\frac{m_Z+m_X}{n}$.

\begin{definition}[Quantum Locally Testable CSS Code (qLTC) {\cite[Fact 17]{AE17,EH17}\protect\footnotemark}]
    A quantum CSS code $(\calC_X,\calC_Z)$ is $s$-locally testable if both $\calC_X$ and $\calC_Z$ are.
\end{definition}

\footnotetext{The definition in \cite{AE17} is more general. With the conventional normalization, \cite[Fact 17]{EH17} showed that a CSS code $(\calC_X,\calC_Z)$ for which both codes are $s$-locally testable is itself $s$-locally testable in the general definition. On the other hand they showed that if the CSS code is $s$-locally testable, then both $\calC_Z,\calC_X$ must be $s/2$-locally testable. In this work we only care about the existence of a constant soundness for quantum local testability so we define the CSS code's local testability soundness directly in terms of the soundness parameters of the two classical codes.}

\subsection{Construction and Proof}
For the quantum locally testable code we will use a $(4,2)$-cubical complex. The code uses the middle three layers as a CSS cochain complex, i.e. $C^1(X)\xrightarrow{\delta_1}C^2(X)\xrightarrow{\delta_2}C^3(X)$, with $\calC_Z = \ker(\delta_2)$ and $\calC_X = \ker(\delta_1^t)$ and we use the coboundary and boundary maps as the parity-check matrices $H_Z\coloneq\delta_2$ and $H_X\coloneq\delta_1^t$. To show this is indeed a qLTC we follow Lin and Hsieh's arguments in \cite{HL22a,HL22b} which used and generalized Sipser and Spielman's decoder of \cite{SS96}. With the adjustments described below, their argument applies in our setting. At a high level, in \cite{HL22b} they show that, using necessary imbalance of the blocks, a $(2,2)$-cubical complex as a CSS code constitutes an asymptotically good quantum LDPC code, and in \cite{HL22a} they show that using such a complex (with different imbalance) and putting the bits in the bottom layer constitutes an asymptotically good classical LTC. We use the former for the sub-complex $C^1(X)\xrightarrow {\delta_1}C^2(X)\xrightarrow{\delta_2}C^3(X)$ to show the code has constant relative distance, and the latter for the sub-complexes $C^2(X)\xrightarrow{\delta_2}C^3(X)\xrightarrow{\delta_3}C^4(X)$ and $C^0(X)\xleftarrow{\delta_0^t}C^1(X)\xleftarrow{\delta_1^t}C^2(X)$ to show each of the codes $\calC_Z$ and $\calC_X$ has a constant local testability soundness. Constant rate is shown by another imbalance condition.

The main generalization that requires an adjustment for our setting is the notion of locality: local minimality in this work requires considering any cochain in the previous level that is supported in any one link of the bottom layer. This is the adjustment we introduce into the decoder algorithm \cite[Definition 16]{HL22a}.

Recall that for $v\in V_0$, a cochain in its link $h\in C^{i-2}(\Lk(v))$ corresponds to the layer $V_{i-1}$ of $X$, and that $\iota_v h \in C^{i-1}(X)$ is its corresponding cochain.

\begin{algorithm}[H]
    \caption{Greedy flipping, adapted from \cite[Def. 16]{HL22a}}
    \label{def:flipping-algo}
    \renewcommand{\algorithmicrequire}{\textbf{Input:}}
    \begin{algorithmic}[1]
        \Require $f\in C^i(X)$, with $1\leq i\leq k$.
        \While{there exist $v\in V_0$ and $h\in C^{i-2}(\Lk(v))$ such that
            $\norm{f+\delta(\iota_vh)}<\norm{f}$}
            \State $f\gets f+\delta(\iota_vh)$
        \EndWhile
        \State \Return $f$
    \end{algorithmic}
\end{algorithm}

\begin{lemma} \label{lem:flipping-algo-local-minimality}
    For an input cochain $f$, \Cref{def:flipping-algo} stops after at most $D_{[k]}\cdot\norm{f}$ iterations and returns a locally-minimal cochain.
\end{lemma}
\begin{proof}
By the stopping condition, no correction supported in a bottom-vertex link decreases the output's norm. Every vertex $v$ has a bottom ancestor $u\in V_0$, with $\Lk(v)\subseteq\Lk(u)$. Thus every correction allowed by \cref{def:minimality} is also supported in a bottom-vertex link, so the output is locally minimal.

Every cochain norm is an integer multiple of $\alpha\coloneq D_{[k]}^{-1}$, since $D_I$ divides $D_{[k]}$ for every $I\subseteq[k]$. Each iteration strictly decreases the norm, hence decreases it by at least $\alpha$. Thus the algorithm stops after at most $\tfrac1\alpha\norm{f}=D_{[k]}\norm{f}$ iterations for an input cochain $f$.
\end{proof}

To obtain constant rate, we require the combined size of the two adjacent layers to be smaller than that of the middle layer by a constant fraction. For that, we will assume a minimal stretch/shrinkage of the direction-graphs in the complex. Define the stretch/shrinkage factor of direction $i$ by $\mu_i\coloneq\frac{D_i}{D'_i}$.

\begin{theorem}[qLTC from a $(4,2)$-cubical complex]\label{cor:qltc}
Let $X$ be a $(4,2)$-cubical complex. Assume $X$ and $X^{\mathrm{rev}}$ satisfy the assumptions in \cref{prop:small-set-cosystolic-expansion} with $\varepsilon_\star$ sufficiently small\footnote{$\varepsilon_\star$ needs to be small so that the inequality in \cref{prop:small-set-cosystolic-expansion} is meaningful.} and with stretch/shrinkage factors  $\mu_1,\mu_2\geq 4$, and $
\mu_3,\mu_4\leq\tfrac14$.

Then $\mathcal{C}(\ker{\delta^t_1},\ker{\delta_2}) $ is a quantum error-correcting code with constant rate, relative distance, and local testability soundness.
\end{theorem}

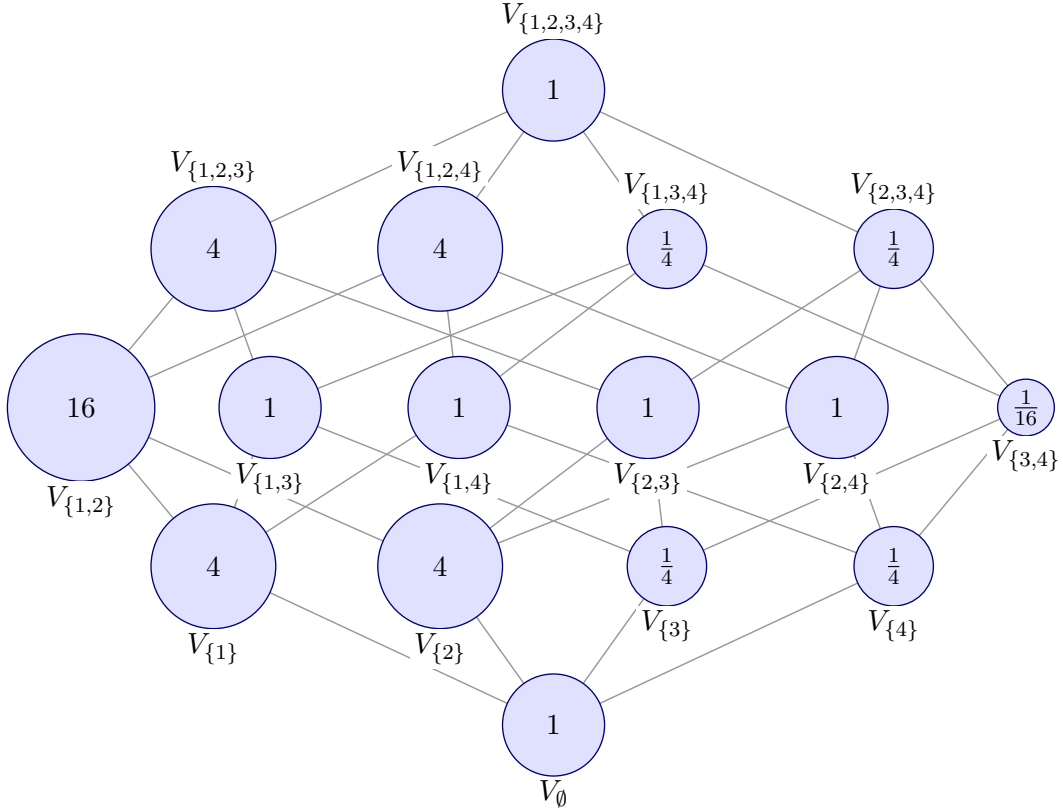
\begin{figure}[h]
        \centering
        \input{figures/boolean-cube-rate}
        \caption{An illustration of the size of the blocks in the four-dimensional Boolean cube for $\mu_1=\mu_2=4$ and $\mu_3=\mu_4=\tfrac14$. Each block shows the ratio $\abs{V_I}/\abs{V_\emptyset}$.}
        \label{fig:boolean-cube-rate}
    \end{figure}
\begin{proof}
    Every vertex in $V_3$ and in $V_1$ has at most a constant number of neighbors in $V_2$. Specifically, every row of $H_Z=\delta_2$ has weight at most $3\max_{i\in[4]}D_i'$, and every row of $H_X=\delta_1^t$ has weight at most $3\max_{i\in[4]}D_i$. Thus the checks have constant weight, at most $3\max_{i\in[4]}\{D_i,D_i'\}$. Similarly, every vertex in $V_2$ has at most $\sum_{i=1}^k D_i+D_i'$ neighbors in both $V_1,V_3$, meaning $H_Z,H_X$ have constant weight columns, or equivalently, every qubit participates in a constant number of checks.

    For the rate, by biregularity, $\abs{V_{I\cup\{i\}}}=\mu_i\abs{V_I}$ for $i\notin I$, so $\abs{V_I}=\bigl(\prod_{i\in I}\mu_i\bigr)\abs{V_\emptyset}$, as illustrated in \cref{fig:boolean-cube-rate}.
    To bound the rate, count qubits and checks
    \[
        \rho \geq  1 - \frac{\abs{V_1} + \abs{V_3}}{\abs{V_2}} > 0.1.
    \]
    To bound both the local testability soundness and the relative distance, we apply the Greedy Flipping Algorithm, \cref{def:flipping-algo}, using both \cref{lem:flipping-algo-local-minimality} and \cref{prop:small-set-cosystolic-expansion}, to show the parameters are bounded below by positive constants. We give the full arguments for both parameters considering the $Z$-codes, $\calC_Z=\ker(\delta_2)$ and $\calC_Z/\calC_X^\perp=\ker(\delta_2)/\im(\delta_1)$, respectively. The same arguments applied to the reversed complex give the local testability soundness and relative distance of the $X$-codes $\calC_X=\ker(\delta^{\mathrm{rev}}_2)=\ker(\delta_1^t)$ and $\calC_X/\calC_Z^\perp=\ker(\delta^{\mathrm{rev}}_2)/\im(\delta^{\mathrm{rev}}_1)=\ker(\delta_1^t)/\im(\delta_2^t)$.

    For distance, let $f\in\ker(\delta_2)\setminus\im(\delta_1)$ represent a nonzero element in $\calC_Z/\calC_X^\perp$. Run the greedy flipping algorithm on it. Its output is $f'=f+\sum_{i=1}^r\delta_1(h_i)$, with $\norm{f'}\leq\norm{f}$, and is locally minimal by \cref{lem:flipping-algo-local-minimality}. Since $f'-f\in\im(\delta_1)=\calC_X^\perp$, the cochains $f$ and $f'$ represent the same nonzero element, so $f'\neq0$. As $\delta_2(f')=0$, \cref{prop:small-set-cosystolic-expansion} gives $\norm{f}\geq\norm{f'}\geq c\norm{V_2}$.

    For the local testability of $\calC_Z=\ker(\delta_2)$, look at its syndrome layer $V_3$. Let $w\in C^2(X)$ be some word (not necessarily a codeword in the kernel of $\delta_2$). By adding the closest codeword to it with no change to its syndrome we may assume its closest codeword is $0$. We wish to show that its norm is bounded above by a constant multiple of the norm of its syndrome. Let $f_{\mathrm{syn}}\coloneq \delta_2(w)$ be its syndrome. Run the greedy flipping algorithm on $f_{\mathrm{syn}}$ and obtain $f_{\mathrm{syn}}'=\delta_2(w)+\sum_{i=1}^r\delta_2(h_i)$ which has at most the same norm and is locally minimal. As $C^2(X)\to C^3(X)\to C^4(X)$ is a cochain complex, $\delta_3(f_{\mathrm{syn}}')=\delta_3(\delta_2(w+\sum_{i=1}^r h_i))=0$. By \cref{prop:small-set-cosystolic-expansion} either $f_{\mathrm{syn}}'=0$ or it has large norm $\norm{f_{\mathrm{syn}}'}\geq c\norm{V_3}$. Divide into these two cases:

    \begin{enumerate}
    \item $\norm{f_{\mathrm{syn}}'}\geq c\norm{V_3}$. The algorithm never increases the norm, so $\norm{f_{\mathrm{syn}}}\geq\norm{f_{\mathrm{syn}}'}$. Also, $\dist(w,\calC_Z)\leq\norm{w}\leq\norm{V_2}$. By biregularity, $\beta\coloneq\norm{V_3}/\norm{V_2}$ is a positive constant depending only on the degree parameters. Thus,
    \[
        \norm{f_{\mathrm{syn}}}
        \geq \norm{f_{\mathrm{syn}}'}
        \geq c\norm{V_3}
        =c\beta\norm{V_2}
        \geq c\beta\cdot\dist(w,\calC_Z).
    \]
    This gives the desired soundness bound with $s=c\beta>0$ in this case.
    \item $f_{\mathrm{syn}}'=0$. Set $u=\sum_{i=1}^r h_i$. Since $\delta_2(w+u)=f_{\mathrm{syn}}'=0$, the corrected word $w+u$ belongs to $\calC_Z$. Each change $h_i\in C^2(X)$ is supported on $\Lk_1(v)$ for some $v\in V_0$. This link level has six blocks, each of ambient norm at most $1$; hence $\norm{h_i}\leq 6$. By \cref{lem:flipping-algo-local-minimality}, $r\leq D_{[k]}\cdot\norm{f_{\mathrm{syn}}}$. Therefore,
    \[
        \dist(w,\calC_Z)
        \leq \norm{u}
        \leq \sum_{i=1}^r\norm{h_i}
        \leq 6r
        \leq {6D_{[k]}}\norm{f_{\mathrm{syn}}}.
    \]
    Rearranging gives $\norm{f_{\mathrm{syn}}}\geq s\cdot\dist(w,\calC_Z)$ with the constant $s=\frac{1}{6D_{[k]}}>0$.
    \end{enumerate}
    Together, the two cases give soundness $s_Z=\min\{c\beta,\frac{1}{6D_{[k]}}\}>0$.
\end{proof}

The construction generalizes to any dimension $k\geq4$ and any three consecutive layers $V_{\ell-1},V_\ell,V_{\ell+1}$ where $2\leq\ell\leq k-2$. To obtain constant rate, it is enough to assume $\mu_i\geq2k$ for $i\leq\ell$ and $\mu_i\leq1/(2k)$ for $i>\ell$. This makes \(V_{[\ell]}\), a block in the middle layer, larger than all other blocks combined by a constant factor.

\subsection{Linear Time Decoder}

Continuing to follow Lin and Hsieh and specifically \cite{HL22b} and Dinur, Hsieh, Lin and Vidick \cite{dinur2023goodquantumldpccodesstoc}, it is possible to implement a version of the greedy flipping algorithm \cref{def:flipping-algo} in linear time. We consider the $Z$-code $\calC_Z/\calC_X^\perp=\ker(\delta_2)/\im(\delta_1)$ here, but the same works for the $X$-code $\calC_X/\calC_Z^\perp=\ker(\delta_1^t)/\im(\delta_2^t)$, using the reversed complex's levels. In order to implement the linear decoder we need to maintain a list of vertices at the bottom layer $Q\subset V_0$, each $v\in Q$ a potential root for a local change, i.e. a cochain in its link $h\in C^{1}(\Lk(v))$ whose addition reduces the current norm of the syndrome. The main observation is that whenever such a cochain is added, only a constant number of new root-vertices that were already examined should be reexamined, while reducing the current syndrome's norm by at least a constant. This makes the algorithm run in linear time in the norm of the syndrome, and thus linear in the number of qubits as well.

For a set of vertices $A\subseteq V$, denote the set of vertices in the bottom layer $V_0$ whose links overlap with $A$ by $\Roots(A)\coloneq\{v\in V_0:\exists u\in A, v\leq u\}$. In the following algorithm, we use this notation to refer to the roots of new possible changes after adding a change $x\in C^3(X)$ to the current syndrome: all roots that may need reexamination belong to $\Roots(\supp(x))$.

\begin{algorithm}[H]
    \caption{Linear time greedy flipping, adapted from \cite[arxiv's Alg. 3]{HL22b, dinur2023goodquantumldpccodesstoc}}
    \label{alg:lin-time-flipping-algo}
    \renewcommand{\algorithmicrequire}{\textbf{Input:}}
    \begin{algorithmic}[1]
        \Require $f\in C^3(X)$.
        \State Initialize $Q\gets \Roots(\supp(f))$; $r\gets0$ 
        \While{$Q\neq\emptyset$ and $f\neq0$}\label{alg:linear-flipping line:while}
        \State Let $v\in Q$ and update $Q\gets Q\setminus\{v\}$
          \For{$h\in C^{1}(\Lk(v))$}
          \If{$\norm{f+\delta(\iota_vh)}<\norm{f}$}
            \State $f\gets f+\delta_2(\iota_vh)$; $h_{r+1}\gets h$ ; $v_{r+1}\gets v$ ; $r\gets r+1$ \label{alg:linear-flipping line:add-change}
            \State $Q\gets Q\cup\Roots(\supp(\delta_2(\iota_v h)))$ \label{alg:linear-flipping line:add-roots}
          \EndIf
          \EndFor
        \EndWhile
        \State \Return $\sum_{i=1}^r \iota_{v_i} h_i$
    \end{algorithmic}
\end{algorithm}

\begin{proposition}
    Let $X$ be a $(4,2)$-cubical complex satisfying the assumptions for \cref{cor:qltc}.

    Denote by $c,\alpha$ the constants in the hypothesis and the conclusion of \cref{prop:small-set-cosystolic-expansion}, respectively. Let $e\in C^2(X)$ be a cochain close to $\ker(\delta_2)$ i.e. $\dist(e,\ker(\delta_2))\leq\min\{\frac{c}{2}\norm{V_3},\frac{c}{2}\norm{V_2}-12,\frac{\alpha}{2}(\frac{c}{2}\norm{V_2}-12)\}$. Then running \cref{alg:lin-time-flipping-algo} on $\delta_2(e)\in C^3(X)$ takes linear time in $\norm{\delta_2(e)}$ (and thus linear in the number of qubits $\abs{V_2}$) and outputs a change $\sum_{i=1}^r h_i$ such that $e+\sum_{i=1}^r h_i\in\ker(\delta_2)$ and there is a boundary $x\in\im(\delta_1)$ such that $e+\sum_{i=1}^r h_i+x$ is the closest cochain in $\ker(\delta_2)$ to $e$.
\end{proposition}

In \cref{prop:small-set-cosystolic-expansion} the conclusion's constant depends on the level for which it is applied. We will be applying it both to level $2$ and $3$, so to make notation easier take $\alpha$ to be the minimum between the two constants.

\begin{proof}
    Let $w\in \ker(\delta_2)$ be a closest one to $e$. We may replace $e$ with $e+w$ without loss of generality, by linearity and by $w$ not changing the syndrome. Then $0$ is a closest word in $\ker(\delta_2)$ to $e$ and we may use the norm as the relevant quantity to be zeroed-out by an equivalent word in $\calC_Z/\calC_X^\perp$ in order to show correctness. 
    
    Throughout this proof we will use notations to talk about the errors and syndromes after several iterations. Denote by $ e'_0\coloneq e$ the initial error and by $e_i'=e+\sum_{j=1}^ig_j$ the error after $i$ iterations of the while loop (line \ref{alg:linear-flipping line:while}), meagning $g_j$ is the sum of the changes $h_\ell$ whose cobaoundary were added during the $j$'th iteration. Denote by $e_i$ the minimal cochain in the class of $e'_i$ , i.e. $e_i\coloneq e'_i+\delta_1(x_i)$ for $x_i$ which minimizes its norm. Clearly $\norm{e_i}\leq\norm{e'_i}$. Denote the syndrome at that time by $f_i\coloneq\delta_2(e_i)=
    \delta_2(e'_i)$. The $e_i$-s are helpful for analysis, but the algorithm does not try to compute them at any point. It only accumulates the $e'_i$-s. 

    First we show the algorithm's linear time. Each iteration of the while loop (line \ref{alg:linear-flipping line:while}) is bounded by a constant number of steps, as the number of cochains in $C^1(\Lk(v))$ is bounded by a constant regardless of the vertex $v\in V_0$. We need to show the algorithm stops after linear number of iterations. $Q$ is initialized to have linear number of items in $\norm{f_0}$, say at most $K_1\norm{f_0}$. In each iteration either a vertex is taken out or the norm of the syndrome is reduced. In the latter case the norm is reduced by at least $D_{[k]}^{-1}$, a constant, and $Q$ is enlarged by at most another constant $K_2=\sum_{\substack{i_1\neq j_1\in[k]\\ i_2\neq j_2\in[k]}}D_{i_1}D_{j_1}D'_{i_2}D'_{j_2}$. Thus, in a linear number of steps (at most $(K_1+D_{[k]}K_2)\norm{f_0}$) we have either an empty $Q$ or a zero $f$, each of which stops the while loop and ends the algorithm.   
    
    Next we show that the syndrome at the end is $0$. Let $t$ be the number of iterations made before the algorithm stopped.
    The syndrome is never increased by the algorithm $\norm{f_t}\leq\norm{f_0}$. Additionally, by \cref{eq:norm-propagation}, $\norm{f_0}\leq2\norm{e'_0}$ and by the proposition's hypothesis $\norm{e'_0}\leq \frac{c}{2}\norm{V_3}$. We will show $f_t$ must be locally minimal, which implies by \cref{prop:small-set-cosystolic-expansion} that it expands $0=\norm{\delta_3(\delta_2(e_t))}=\norm{\delta_3(f_t)}\geq\alpha\norm{f_t}$ showing $f_t=0$. To see $f_t$ is locally minimal note the algorithm, before each iteration at line \ref{alg:linear-flipping line:while}, conserves the invariant: for every cochain $h\in C^2(X)$ that can reduce the syndrome $\norm{f_i+\delta_2(h)}<\norm{f_i}$ and is supported inside a single link $C^1(\Lk(v_h))$ for $v_h\in V_0$, the vertex $v_h$ is in $Q$.\footnote{There can be more than one such vertex -- for each such cochain assign one arbitrarily} This invariant clearly holds just before the first iteration, by the initialization. Assume the invariant holds before the $i$'th iteration. During this iteration the algorithm takes out one $v_i\in Q$ and examines all the cochains in $V_2$ supported inside its link. $f_i$ is the syndrome at the end of this iteration. Let $h'\in C^2(X)$ be a cochain that reduces $f_i$'s norm and that is supported on a single link $\Lk(v_{h'})$. Either $\delta_2(h')$ overlaps with the coboundary $\delta_2(h)$ of one of the cochains $h$ that were added to the overall change in line \ref{alg:linear-flipping line:add-change} during this iteration, or it does not overlap with any of them (including the case of having none during this iteration). If it does, then $v_{h'}\in\Roots(\supp(\delta_2(h)))$ and $v_{h'}$ was added to $Q$ in line \ref{alg:linear-flipping line:add-roots} during the iteration. Otherwise, its coboundary does not overlap with any bit that was changed from $f_{i-1}$ to $f_i$ and hence must decrease the norm of $f_{i-1}$ as well. In this case $v_{h'}\in Q$ before this iteration and $v_{h'}\neq v_i$ so it was not taken out from $Q$, showing it to still be in $Q$ after the iteration. 
    
    We will use this now to show the error at the end is also $0$. By regularity and the normalization of the norm, $\norm{f_i}\leq2\norm{e_i}$. On the other hand, for $\norm{e_i}\leq c\norm{V_2}$ we get by \cref{prop:small-set-cosystolic-expansion} $\norm{e_i}\leq \frac1\alpha\norm{f_i}$ for a constant $\alpha$. By induction we'll see the condition is satisfied for all $i$, showing the final error to be $0$. The induction will show and use something stronger: that $\norm{e_i}\leq \frac{c}{2}\norm{V_2}-\beta$ for the constant  $\beta\coloneq 12$ bounding the norm of all vertices in $V_2$ in a single link ($1$ for each monotone path of directions from $V_0$ to $V_2$, i.e. for each ordered pair $i\neq j\in[4]$), which is the maximal change in the error in a single iteration. The induction basis is given by the hypothesis of the proposition. For the step, by the induction hypothesis $\norm{e_{i-1}}\leq \frac{c}{2}\norm{V_2}-\beta$ and the fact that at most $\beta$ bits are changed from $e_{i-1}$ to $e_{i-1}+g$ for $g\in C^2(X)$ supported in a single link, which is equivalent to the minimal $e_i$, we get that $\norm{e_i}\leq\norm{e_{i-1}+g}\leq \frac{c}{2}\norm{V_2}$. Thus by \cref{prop:small-set-cosystolic-expansion} we have $\norm{e_i} \leq \frac1\alpha\norm{f_i} \leq \frac1\alpha\norm{f_0} \leq \frac2\alpha\norm{e_0} \leq \frac2\alpha\frac{\alpha}{2}(\frac{c}{2}\norm{V_2}-\beta)=\frac{c}{2}\norm{V_2}-\beta$ where the second inequality is by the norm of the syndrome never increasing in iterations, and the last inequality is by the proposition's hypothesis.  
\end{proof}

%% file: figures/boolean-cube-rate.tex
% The Boolean 4-cube for mu_1 = mu_2 = 4, mu_3 = mu_4 = 1/4.
% Disk sizes indicate the ordering of relative block sizes qualitatively.
\begin{tikzpicture}[
    block/.style={circle, fill=blue!12, draw=blue!45!black,
        line width=.5pt, inner sep=0pt, outer sep=0pt},
    block label/.style={fill=white, inner sep=1.5pt, align=center,
        font=\small, anchor=north},
    cube edge/.style={draw=black!40, line width=.55pt}
]
    \coordinate (v0) at (0,0);
    \coordinate (v1) at (-4.5,2.1);
    \coordinate (v2) at (-1.5,2.1);
    \coordinate (v3) at (1.5,2.1);
    \coordinate (v4) at (4.5,2.1);
    \coordinate (v12) at (-6.25,4.2);
    \coordinate (v13) at (-3.75,4.2);
    \coordinate (v14) at (-1.25,4.2);
    \coordinate (v23) at (1.25,4.2);
    \coordinate (v24) at (3.75,4.2);
    \coordinate (v34) at (6.25,4.2);
    \coordinate (v123) at (-4.5,6.3);
    \coordinate (v124) at (-1.5,6.3);
    \coordinate (v134) at (1.5,6.3);
    \coordinate (v234) at (4.5,6.3);
    \coordinate (v1234) at (0,8.4);

    % Draw the 32 edges first so that labels remain unobstructed.
    \foreach \a/\b in {
        0/1,0/2,0/3,0/4,
        1/12,1/13,1/14,2/12,2/23,2/24,
        3/13,3/23,3/34,4/14,4/24,4/34,
        12/123,12/124,13/123,13/134,14/124,14/134,
        23/123,23/234,24/124,24/234,34/134,34/234,
        123/1234,124/1234,134/1234,234/1234}
        \draw[cube edge] (v\a) -- (v\b);

    % Diameters range from 0.75 to 1.95 cm, in steps of 0.30 cm per factor of 4.
    % Ratios are inside the disks; block names are above the upper two ranks
    % and below the remaining ranks.
    \foreach \id/\subset/\ratio/\value in {
        0/\emptyset/1/1,
        1/{\{1\}}/4/4, 2/{\{2\}}/4/4,
        3/{\{3\}}/{\frac14}/{1/4}, 4/{\{4\}}/{\frac14}/{1/4},
        12/{\{1,2\}}/16/16,
        13/{\{1,3\}}/1/1, 14/{\{1,4\}}/1/1,
        23/{\{2,3\}}/1/1, 24/{\{2,4\}}/1/1,
        34/{\{3,4\}}/{\frac{1}{16}}/{1/16},
        123/{\{1,2,3\}}/4/4, 124/{\{1,2,4\}}/4/4,
        134/{\{1,3,4\}}/{\frac14}/{1/4},
        234/{\{2,3,4\}}/{\frac14}/{1/4},
        1234/{\{1,2,3,4\}}/1/1}
    {
        \pgfmathsetmacro{\diameter}{1.35 + 0.30*ln(\value)/ln(4)}
        \node[block, minimum size=\diameter cm] (b\id) at (v\id) {};
    }
    % Place labels after all disks, above the connecting edges.
    \foreach \id/\subset/\ratio in {
        0/\emptyset/1,
        1/{\{1\}}/4, 2/{\{2\}}/4,
        3/{\{3\}}/{\frac14}, 4/{\{4\}}/{\frac14},
        12/{\{1,2\}}/16,
        13/{\{1,3\}}/1, 14/{\{1,4\}}/1,
        23/{\{2,3\}}/1, 24/{\{2,4\}}/1,
        34/{\{3,4\}}/{\frac{1}{16}},
        123/{\{1,2,3\}}/4, 124/{\{1,2,4\}}/4,
        134/{\{1,3,4\}}/{\frac14}, 234/{\{2,3,4\}}/{\frac14},
        1234/{\{1,2,3,4\}}/1}
        {
            \node[font=\small, inner sep=0pt] at (v\id) {$\ratio$};
            \ifnum\id>100
                \node[block label, anchor=south] at (b\id.north) {$V_{\subset}$};
            \else
                \node[block label] at (b\id.south) {$V_{\subset}$};
            \fi
        }
\end{tikzpicture}